\pdfoutput=1
\documentclass[12pt,draftclsnofoot, onecolumn]{IEEEtran}
\usepackage{amsthm}
\usepackage{graphicx}
\usepackage{subfigure}
\usepackage{amsmath}
\usepackage{amssymb}%
\usepackage{multicol}
\usepackage{float}
\usepackage{color}
\usepackage{array}
\usepackage{caption}
\usepackage{multirow}
\usepackage{algorithm}
\usepackage{algorithmic}
\usepackage{epstopdf}
\usepackage{cases}
\usepackage{url}

\usepackage{ragged2e}
\usepackage{multirow}
\usepackage[table]{xcolor}
\usepackage[numbers]{natbib}

\makeatletter

\newcommand{\Rmnum}[1]{\expandafter\@slowromancap\romannumeral #1@}
\makeatother

\newtheorem{theorem}{Theorem}
\newtheorem{lemma}{Lemma}
\newtheorem{proposition}{Proposition}

\newtheorem{corollary}{Corollary}

\usepackage[T1]{fontenc}
\usepackage{ifpdf}
 \ifpdf
 \else
 \fi

\begin{document}
%
\title{Energy-saving deployment algorithms of UAV swarm for sustainable wireless coverage}

\author{ Xiao~Zhang,~\IEEEmembership{Member,~IEEE,} and
         Lingjie Duan,~\IEEEmembership{Senior Member,~IEEE}
\thanks{

X. Zhang is with College of Computer Science, South-Central University for Nationalities. L. Duan are with Engineering Systems and Design Pillar, Singapore University of Technology and Design, Singapore. E-mail: xiao.zhang@my.cityu.edu.hk, lingjie\_duan@sutd.edu.sg. Part of the work by X. Zhang was done when he was a postdoc research fellow at Singapore University of Technology and Design.

} }

\maketitle


\begin{abstract}
Recent years have witnessed increasingly more uses of Unmanned Aerial Vehicle (UAV) swarms for rapidly providing wireless coverage to ground users. Each UAV is constrained in its energy storage and wireless coverage, and it consumes most energy when flying to the top of the target area, leaving limited leftover energy for hovering at its deployed position and keeping wireless coverage. The literature largely overlooks this sustainability issue of deploying UAV swarm deployment, and we aim to maximize the minimum leftover energy storage among all the UAVs after their deployment. Our new energy-saving deployment problem captures that each UAV's wireless coverage is adjustable by its service altitude, and also takes the no-fly-zone (NFZ) constraint into account. Despite of this, we propose an optimal energy-saving deployment algorithm by jointly balancing heterogeneous UAVs' flying distances on the ground and final service altitudes in the sky. We show that a UAV with larger initial energy storage in the UAV swarm should be deployed further away from the UAV station. Moreover, when $n$ homogeneous UAVs are dispatched from different initial locations, we first prove that any two UAVs of the same initial energy storage will not fly across each other, and then design an approximation algorithm of complexity $n \log \frac{1}{\epsilon}$ to arbitrarily approach the optimum with error $\epsilon$. Finally, we consider that UAVs may have different initial energy storages, and we prove this problem is NP-hard. Despite of this, we successfully propose a heuristic algorithm to solve it by balancing the efficiency and computation complexity well.
\end{abstract}


\IEEEpeerreviewmaketitle

%

\vspace{3cm}

\IEEEraisesectionheading{\section{Introduction}\label{sec:intr}}
\IEEEPARstart{R}ecently, there are increasingly more exercises and commercial uses of Unmanned Aerial Vehicle (UAV) swarms for rapidly providing wireless coverage to ground users (e.g., \cite{zeng2016wireless} \cite{zhang2017optimization} \cite{chenjimingUAV} \cite{xu2018uav}). In these applications, UAVs serve as flying base stations to serve a geographical area (e.g., cell edge or disaster zone) out of the capacity or reach of territorial base stations. The continuing development of UAV applications for keeping wireless coverage faces two key challenges. First, since each UAV's wireless coverage radius (though adjustable by its deployed altitude) is small, it consumes most energy when flying over a long distance to reach the target area to serve ground users closely. This leaves limited leftover energy for the UAV network's lifetime of keeping wireless coverage afterwards, results in a severe sustainability issue. The endurance of each UAV's on-board energy storage is fundamentally limited by its weight and aircraft size. Multiple UAVs in the swarm need to cooperate to balance their energy consumption during deployment to fully cover users in the distant target area.

Second, many countries have set up sizable No-Fly-Zones (NFZs) which prohibit any deployment of UAVs inside \cite{NFZ}. Usually, NFZs include restricted areas, prohibited areas, and military bases. Take Singapore as a typical urban city example, Figure~\ref{fig:NFZ} shows that NFZs in orange widely cover airports, air bases and military context~\footnote{https://garuda.io/what-you-must-know-about-drone-no-fly-zones-nfz/}. UAVs can only fly at very low altitude to cross these NFZs before reaching their final deployment positions outside NFZs. The optimal UAV swarm deployment should consider these NFZ constraints. To the best of our knowledge, none of the existing work study the energy-saving UAV swarm deployment algorithms by considering the UAV-to-UAV cooperation and NFZ constraint. Further, since each UAV's wireless coverage is adjustable by its hovering altitude, we should jointly optimize its flying distance on the ground and final service altitude in the sky, which was also overlooked in the literature.


\begin{figure}[t]
    \centering
        \includegraphics[width=0.7\textwidth]{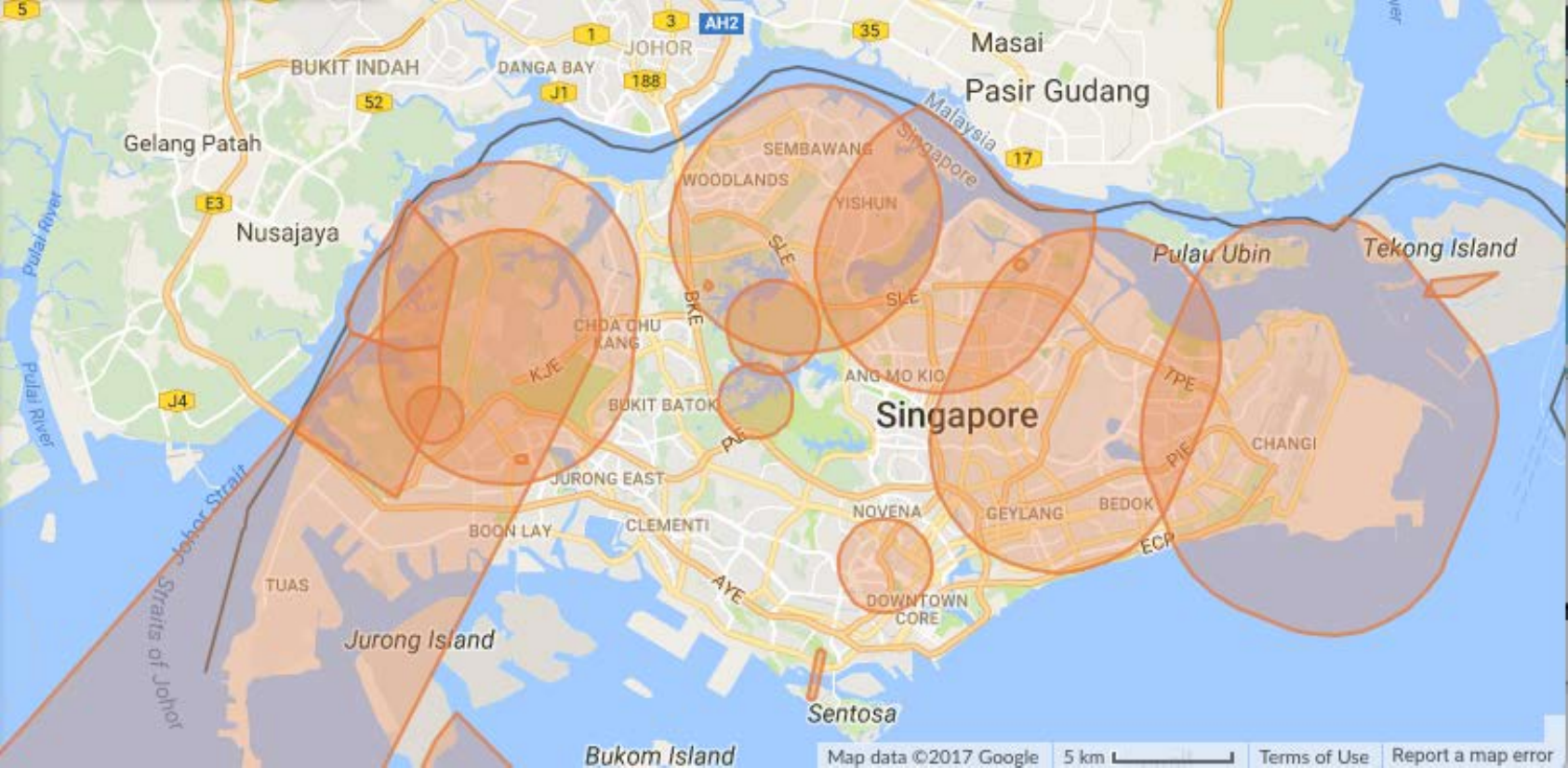}
    \caption{The distribution of No-Fly-Zones (NFZs) in orange in Singapore.}
    \label{fig:NFZ}
\end{figure}

\begin{table*}[tbh]%
\centering
\scriptsize
\caption{Summary of our algorithms for efficiently deploying the UAV swarm in different scenarios.}
\label{tab:algs}%
\begin{tabular}{|l|l|l|l|l|}
\hline
\rowcolor{lightgray}
Five deployment scenarios                                   & Performance   & Complexity    & Place in the paper  \\ [3pt]
\hline
\hline
Same initial locations and energy storages among UAVs   & optimal  & $O(1)$         & Section~\ref{sec:sameb} \\[3pt]\hline
Same initial locations \& different initial energy storages among UAVs      & optimal       & $O(n \log n)$  & Algorithm \ref{alg:same} in Section~\ref{sec:zone} \\[3pt]\hline
Different initial locations \& same energy storages & (1+$\epsilon$)-approximation & $n \log \frac{1}{\epsilon}$ & Algorithm \ref{alg:fptas} in Section \ref{sec:different}\\[3pt]\hline
Different initial locations \& different energy storages &   near optimal    & $O(n \log \frac{1}{\epsilon} C_{n}^{\kappa} \kappa !)$ & Algorithm \ref{alg:Heuristic} in Section \ref{sec:diffb} \\[3pt]\hline
Deployment in 3D space with different energy storages   & (1+$\epsilon$)-approximation & $n \log \frac{1}{\epsilon}$ & Algorithm \ref{alg:FPTAS2} in Section \ref{sec:3D}\\\hline
\end{tabular}
\end{table*}

The literature of UAV-enabled wireless communications assumes UAVs are already in or around the target area to serve ground users, and overlooks the energy-saving issue during the deployment phase of UAVs to reach the target (e.g., \cite{pan2012cooperative, al2014optimal, mozaffari2015drone, zeng2017energy, zhang2018fast}). Recent work on UAV-enabled communications have studied multiple issues such as air-to-ground transmission modeling \cite{al2014optimal} \cite{mozaffari2015drone}, interference management \cite{liwangTVT} \cite{azari2017coverage}, and UAV trajectory planning \cite{wang2018unmanned}. For example, \cite{al2014optimal} and \cite{mozaffari2015drone} investigate the optimal service altitude for a single UAV, where a higher service altitude of the UAV increases the line-of-sight opportunity for air-to-ground transmission but incurs a larger path loss. \cite{azari2017coverage} studies the mutual interference of a UAV downlink links and analyzes the link coverage probability between the UAV and ground users. \cite{wang2018unmanned} uses a UAV-enabled base station to serve multiple users on the ground and jointly optimize the transmit power and the UAV trajectory to maximize the average throughput per ground user. \cite{zeng2017energy} and \cite{wu2017joint} study the energy-efficient UAV movement and UAV-user link scheduling when serving users, and  \cite{wang2018traffic} studies how a UAV should dynamically adapt its location to user movements. Due to a UAV's small wireless service coverage, it consumes most energy when flying over a long distance to the top of the target area, leaving limited leftover energy for the UAV network's hovering and wireless coverage in service phase. It is important to optimize the UAV network deployment before the actual service phase, yet this sustainable deployment issue is largely overlooked in the literature. There are very few works studying the network deployment phase (\cite{xu2018uav} \cite{zhang2018fast}). Further, \cite{xu2018uav} studies the UAV-user interaction for learning users' truthful locations from strategic users before UAV deployment. \cite{wang2018dynamic} studies the economics issues (e.g., pricing and energy allocation) for deploying UAV-provided services.


In this paper, we study energy-saving deployment of a UAV swarm to prolong the UAV swarm's residual lifetime for keeping wireless coverage after deployment. We present Table~\ref{tab:algs} to summarize our proposed algorithms for various deployment scenarios happening in practice. Our key novelty and main contributions are summarized as follows.
\begin{itemize}
\item \emph{Energy-saving deployment of UAV swarm under UAVs' cooperation and NFZ constraint (Section~\ref{sec:sys}):} To our best knowledge, this is the first paper to study the energy-saving issue for deploying a UAV swarm to fully cover a target area, and we aim to provide long enough UAV-provided wireless coverage by seeking UAVs' mutual cooperation. We jointly optimize multi-UAVs' flying distances on the ground and service altitudes in the sky for energy saving purpose, by practically considering the correlation between each UAV's service altitude and its coverage radius as well as the NFZ distribution Our objective is to maximize the UAV swarm's lifetime which is defined as the minimum leftover energy storage among all the UAVs after their deployment.

\item \emph{Optimal deployment by balancing multi-UAVs' energy consumptions in their flights (Section~\ref{sec:same}):} When UAVs are initially located in the same location (e.g., the closest UAV station), we first propose an optimal deployment algorithm of constant complexity $O(1)$ without considering NFZ, and show that a UAV with larger initial energy storage should be deployed further away on the ground for balancing multi-UAVs' energy consumptions in the flights. Further, we show the NFZ constraint will make our problem complex and we manage to present an optimal algorithm of complexity $O(n \log n)$. We show some UAVs will be moved out of the NFZ to the edges, reducing the lifetime of the UAV swarm.

\item \emph{Near-optimal UAV deployment from different initial locations:} In Section~\ref{sec:ddif}, when dispatching UAVs from different initial locations, we first prove that any two UAVs of the same energy storage should not fly across each other during the deployment. This greatly helps us simplify the UAV swarm deployment problem by keeping the UAVs' final position order along the ground. Then we successfully design an ($1-\epsilon$)-approximation algorithm of time complexity $O(n \log \frac{1}{\epsilon})$. Further, when UAVs are of different initial locations, we prove the problem is NP-hard and propose a heuristic algorithm to balance the performance efficiency and computational complexity well. Finally, in Section V, we extend the swarm deployment to a 3D space where the target area to cover is no longer in 1D ground line but in the 2D ground plane.
\end{itemize}

\begin{figure}[t]
    \centering
        \includegraphics[width=0.5\textwidth]{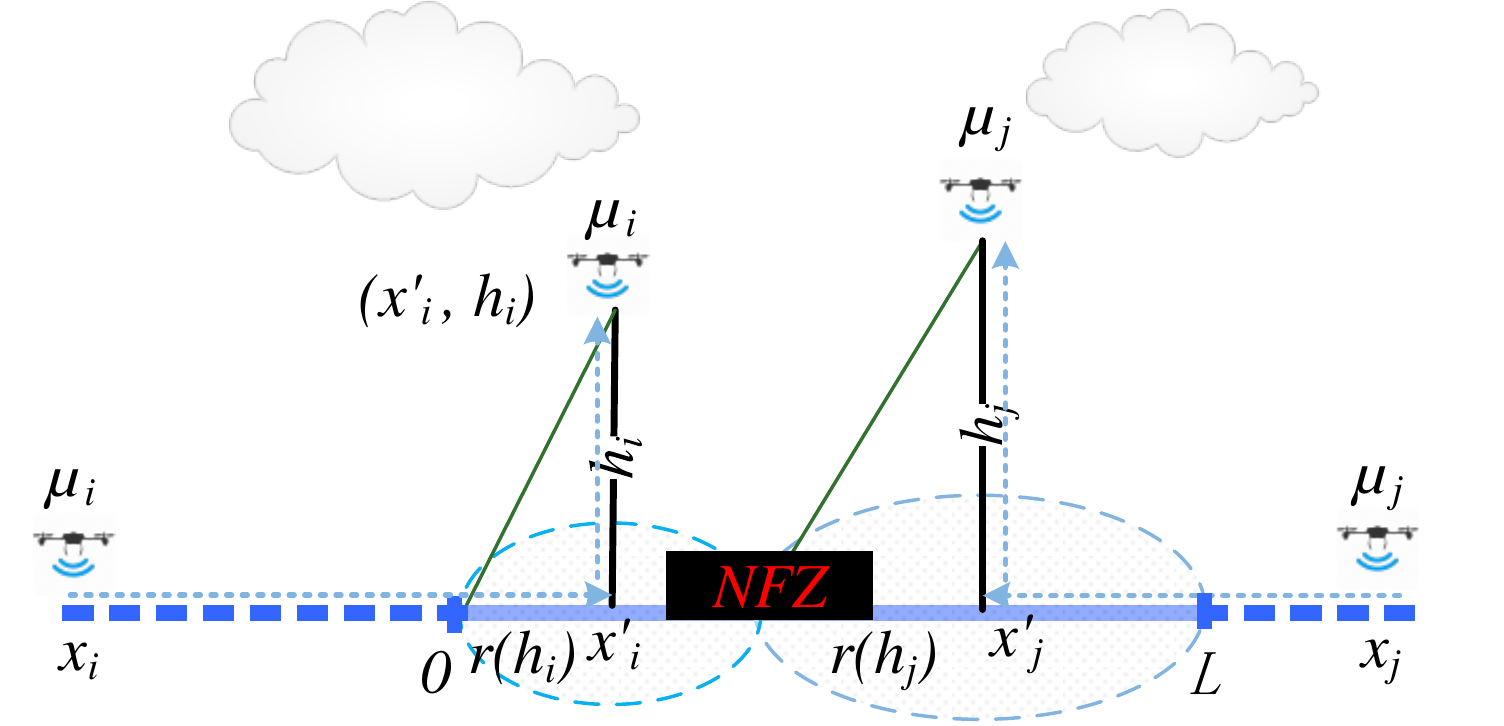}
    \caption{An illustrative example of deploying $n=2$ UAVs from their initial locations $x_i$ and $x_j$ to provide wireless coverage to the whole target area $[0, L]$. Here UAV $\mu_i$ with energy storage $B_i$ is deployed from $x_i$ initially to $x_i^{\prime}$ at service altitude $h_i$ with coverage radius $r(h_i)$ as a function of $h_i$. Neither $x_i^{\prime}$ or $x_j^{\prime}$ can be deployed inside the NFZ within $[0, L]$.}
    \label{fig:covermodel}
\end{figure}

\section{System Model and Problem Formulation}\label{sec:sys}

This section introduces our system model and problem formulation for deploying multi-UAVs to provide full wireless coverage to a target service area on the ground. The target area includes potential users in an activity to be served (e.g., crowd celebrating new year in the fifth avenue in Manhattan) and we first model the target area as a line interval $\textbf{L} = [0, L]$ in 1D, as shown in Figure~\ref{fig:covermodel}. Here, a number $n$ of UAVs in a set $\textbf{U} = \{\mu_1, \cdots, \mu_n\}$ are initially rested in 1D ground locations $\{x_1, \cdots, x_n\}$ with initial energy storages $\{B_1, \cdots, B_n\}$ before their deployment. Later in Section~\ref{sec:3D}, we will extend to model the target area in 2D (see Figure~\ref{fig:covermodel3Dswap}) and generalize our design of UAV deployment algorithms to a 3D space by considering altitude. Here, denote any UAV $\mu_i$'s final position after deployment as $(x_i^{\prime}, h_i)$ at ground distance $x_i^{\prime}$ and hovering altitude $h_i$. After deployment, UAV $\mu_i$ at position $(x_i^{\prime}, h_i)$ covers $[x_i^{\prime}-r(h_i), x_i^{\prime}+r(h_i)]$ in $[0, L]$. We require a full coverage over the target interval $[0, L]$ by deploying $n$ cooperative UAVs, i.e., $[0, L] \subseteq \bigcup_{i=1}^n [x_i^{\prime} - r(h_i), x_i^{\prime} + r(h_i)]$. According to the air-to-ground transmission model in \cite{al2014optimal} and \cite{mozaffari2015drone}, the wireless coverage radius of UAV $r(h_i)$ on the ground concavely increases with the service altitude $h_i$ due to LoS benefit and then decreases due to path loss from the turning point $h^*$. For energy saving, a UAV's flying to an altitude higher than $h^*$ is not necessary for deployment since it will not enlarge the coverage range but consume more energy for flying longer. Therefore, at the optimum the altitude $h_i$ of each UAV $\mu_i$ after deployment to be not higher than $h^*$.



Due to the NFZ policy, a UAV cannot be finally deployed within NFZ. Suppose there are in total $K$ NFZs in this area, and we model NFZ $k$ as a sub-interval $\delta_k = [\delta_k^l,  \delta_k^r]$ in $[0, L]$. For UAV $\mu_i$, $x_i^{\prime}\not\in \cup_{k=1}^K \delta_k$. Despite of this, it is still allowed for UAVs to just bypass NFZs on the way at low altitude according to \cite{NFZ}. To bypass NFZs, Figure~\ref{fig:covermodel} shows that UAV $\mu_i$ first flies horizontally from $x_i$ to $x_i^{\prime}$, then flies vertically up to $h_i$\footnote{This is equivalent to the deployment case that UAVs fly vertically first to the desired altitudes and then horizontally to the destination.}. It travels a normalized distance $d_i(x_i^{\prime}, h_i) = w \cdot |x_i-x_i^{\prime}|+h_i$, where $w < 1$ tells the different energy consumptions per unit horizontally and vertically flying distances. In practice, it is more energy consuming to fly vertically to the sky than horizontally along the ground. One may wonder the relationship between the normalized flying distance $d_i$ and the energy consumption for a UAV. According to \cite{figliozzi2017drones}, the energy consumption of UAV $\mu_i$ is a linear term $c \dot d_i$. The value of energy consumption $c$ per unit distance depends on UAV prototype, and is estimated as $21.6$ Watt hour per km (Wh/km) for UAV prototype \emph{MD4-3000} and $10.8$ Wh/km for prototype \emph{DJI S1000}, respectively.

After deployment, UAV $\mu_i$ only has leftover energy $B_i - c \cdot d_i$ to hover and providing wireless coverage. To prolong the whole UAV swarm's lifetime, we aim to maximize the minimum leftover energy among all the UAVs. Once one UAV uses up its energy, we can no longer guarantee full wireless coverage over $[0,L]$ and the UAV swarm's service lifetime is ended up. Our energy-saving deployment problem of the UAV swarm is
{
\setlength{\abovedisplayskip}{4pt}
\setlength{\belowdisplayskip}{1pt}
\begin{align}
\label{general}
&~ \max_{ \{(x_1^{\prime},h_1), \cdots, (x_n^{\prime},h_n)\}} \underset{1 \leq i \leq n} {\min} ~ B_i - c \cdot {d}_i(x_i^{\prime}, h_i) ~ , \\
&~ {\rm s.t.}, \ [0, L] \subseteq  \bigcup_{i=1}^n [x_i^{\prime} - r(h_i), x_i^{\prime} + r(h_i)],\notag \\
& ~~~~~~~~~ x_i^{\prime} \not\in \cup_{k=1}^K \delta_k. \notag
\end{align}}
\noindent To solve this max-min problem requires $n$ UAVs to cooperate with each other in deployment distance and altitude to evenly use up their energy finally. The problem~(\ref{general}) belongs to the domain of combinatorial optimization and the deployment solution is a specific combination of ordered UAVs above the ground in 2D, which is generally exponential in the number of UAVs and difficult to solve.

\section{UAV swarm deployment from a co-located UAV Station} \label{sec:same}

In this section, we study problem (\ref{general}) when dispatching all the $n$ UAVs from the nearest UAV ground station  (i.e., $x_i = x_j$ for $1 \leq i, j \leq n$). This is the case for covering a not huge service area $L$, and we do not need more UAVs from some other distant UAV stations. Without loss of generality, we assume that $x_i = 0$, $\forall 1 \leq i \leq n$, which is symmetric to the case of $x_i = L$. Then the total travel distance of UAV $\mu_i$ is $d_i=w x_i^{\prime}+h_i$.  Note that for the case of $0 < x_i < L$, we can divide the line interval into two subintervals (i.e., $[0, x_i]$ and $[x_i, L]$), and apply our deployment algorithm (as presented later) similarly over both subintervals. We first skip the NFZ constraint in the first subsection as a benchmark and will add it back later to tell its effect.

\subsection{Deployment of UAVs without NFZ constraint} \label{sec:sameb}

During the deployment, UAVs should cooperate to cover the whole target area and balance their energy consumptions. We have the following result for multi-UAV cooperation.

\begin{proposition}\label{seamless}
At an optimal solution to problem (\ref{general}), all the UAVs have the same amount of leftover energy storage after deployment, i.e., $B_i-c d_i=B_j-c d_j$, for $1\leq i, j\leq n$. Their coverage radii do not overlap with each other and seamlessly cover the target interval. That is, $[0, L] =  \bigcup_{i=1}^n [x_i^{\prime} - r(h_i), x_i^{\prime} + r(h_i)]$, as illustrated in Figure~\ref{fig:attached}.
\end{proposition}
\begin{proof}
See Appendix \ref{ap:att}.
\end{proof}

As the objective of problem~(\ref{general}) is to maximize the minimum leftover energy bottleneck among the UAVs, we do not expect any two UAVs ending with different levels of leftover energy storage after deployment. Otherwise, we can improve the performance by lowering the UAV with less leftover energy and increasing the other UAV's altitude to make up the coverage gap. Then we have the following corollary. When the initial energy storage is identical, we expect the same flying distance to keep the same leftover energy among UAVs. The more energy is consumed in flying horizontal distance, less energy is left for flying vertically up to service altitude. Therefore, a UAV deployed further away on the ground should be placed to a lower altitude, while a closer UAV should be placed to a higher altitude for balancing multi-UAV energy consumption during deployment. It should be noted that if UAVs have different initial storages, the result below may not hold.

\begin{corollary}\label{coro:same}
In the special case that UAVs have the same initial energy storages (i.e., $B_i=B_j$~for~$1 \leq i,j \leq n$), in the optimal deployment solution, UAV $\mu_i$ deployed further away (i.e., with large $x_i^{\prime}$ on the ground) should be placed to a lower altitude $h_i$ for keeping the same energy consumptions among the UAVs during the deployment.
\end{corollary}

\begin{figure}[t]
    \centering
        \includegraphics[width=0.7\textwidth]{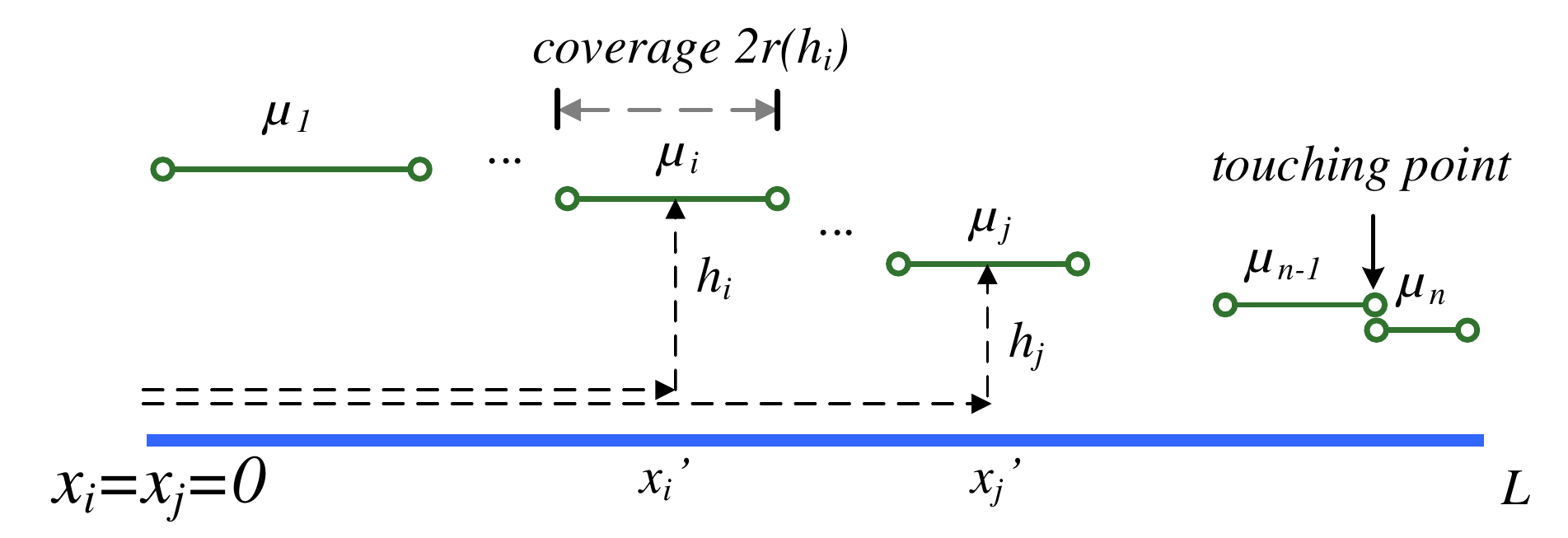}
    \caption{In the optimal solution, all the UAVs have the same leftover energy and their non-overlapping coverages seamlessly reach full coverage of target area $[0,L]$.}
    \label{fig:attached}
\end{figure}

If UAVs are heterogeneous to have different initial energy storages upon deployment, we have following proposition to show that the UAV with larger energy storage should be dispatched further away along the ground in the optimal solution.

\begin{proposition}\label{border}
Without the loss of generality, suppose the initial energy storages of the $n$ UAVs satisfy $B_1 \leq B_2 \leq \cdots \leq B_n$. Then the ground destinations of UAVs satisfy $x_1^{\prime} \leq x_2^{\prime} \leq  \cdots \leq x_n^{\prime}$ in the optimal solution to problem (\ref{general}).
\end{proposition}
\begin{proof}
First of all, for energy saving, a UAV's flying to an altitude higher than $h^*$ is not necessary for deployment since it will not enlarge the coverage range but consume more energy for flying longer. Therefore, at the optimum the altitude $h_i$ of each UAV $\mu_i$ after deployment to be not higher than $h^*$. Thus, in our algorithm design, we only need to consider the concavely increasing part of the function $r(h_i)$.

At the optimum, we denote each UAV's leftover energy storage after deployment as $\hat{B}$. Consider two neighboring UAVs $\mu_i$ and $\mu_j$ with initial energy storages $B_i$ and $B_j$ satisfying $B_i < B_j$. They seamlessly cover a continuous interval $[x_j^{\prime} - r(h_j), x_i^{\prime} + r(h_i)]$ with touching point $x_j^{\prime} + r(h_j) = x_i^{\prime} - r(h_i)$, as shown in the upper subfigure of Figure~\ref{fig:hold}. We have $\hat{B}_i = B_i - c \cdot (w \cdot x_i^{\prime} + h_i) = \hat{B}$, and  $\hat{B}_j = B_j - c \cdot (w \cdot x_j^{\prime} + h_j) = \hat{B}$ for leaving the same residual energy after deploying the two UAVs. We prove by contradiction by supposing $x_i^{\prime} > x_j^{\prime}$ at optimality given $B_i < B_j$, then we have $h_i < h_j$. As illustrated in the lower subfigure of Figure~\ref{fig:hold}, we then swap $\mu_i$ and $\mu_j$ to show a better solution is actually achieved. Specifically, we move $\mu_i$ to $x_i^{\prime\prime}$ at altitude $h_i^{\prime}$ such that $x_i^{\prime\prime} + r(h_i^{\prime}) = x_j^{\prime} - r(h_j)$ to cover the same starting point in the target area. We can see that $\mu_i$ covers $[x_i^{\prime\prime} - r(h_i^{\prime}), x_i^{\prime\prime} + r(h_i^{\prime})]$. Then we divide our discussion, depending on the relationship between $x_i^{\prime\prime}$ and $x_j^{\prime}$.

If $x_i^{\prime\prime} \geq x_j^{\prime}$, then $h_i^{\prime} \geq h_j$ due to larger coverage and $x_i^{\prime\prime} + r(h_i^{\prime}) \geq x_j^{\prime} + r(h_j)$. In this case, we can simply move $\mu_j$ to $x_j^{\prime\prime} = x_i^{\prime}$ and $h_j^{\prime} = h_i$ to cover prior $[x_j^{\prime} + r(h_j), x_i^{\prime} + r(h_i)]$, as UAV $\mu_j$ has larger energy storage at $x_i^{\prime}$ than UAV $\mu_i$ does at $x_i^{\prime}$. Since $x_i^{\prime\prime} + r(h_i^{\prime}) \geq x_j^{\prime} + r(h_j)$ now, these two UAVs unnecessarily overlap in their coverage and we can further improve this solution beyond the optimal solution (before UAVs' swapping). This completes our proof by contradiction for this case.

Next, we consider $x_i^{\prime\prime} < x_j^{\prime}$, then $h_i^{\prime} < h_j$ and we move UAV $\mu_j$ rightwards to $x_j^{\prime\prime}$ such that $x_j^{\prime\prime} - r(h_j^{\prime}) = x_i^{\prime\prime} + r(h_i^{\prime})$ for seamless coverage from UAV $\mu_i$. As shown in Figure~\ref{fig:hold}, $x_i^{\prime} - 2r(h_j) = r(h_i)$, $x_j^{\prime\prime} - 2r(h_i^{\prime}) = r(h_j^{\prime})$. If $x_j^{\prime\prime} \geq x_i^{\prime}$ given $h_i^{\prime} < h_j$, we have $x_j^{\prime\prime} - 2r(h_i^{\prime}) = r(h_j^{\prime}) \geq x_i^{\prime} - 2r(h_i^{\prime}) \geq x_i^{\prime} - 2r(h_j) = r(h_i)$, implying $r(h_j^{\prime})\geq r(h_i)$. Therefore, we have $x_j^{\prime\prime} + r(h_j^{\prime}) \geq x_i^{\prime} + r(h_j^{\prime}) \geq x_i^{\prime} + r(h_i)$. By using up the same amount of energy for both UAVs, we cover a larger total coverage than the optimal solution, telling that UAV swapping provides a better solution than the assumed optimal solution. This completes our proof by contradiction for this subcase.

Now, we only need to consider the other subcase of $x_j^{\prime\prime} < x_i^{\prime}$. As $h_i^{\prime} = \frac{B_i - \hat{B}}{c} - w \cdot x_i^{\prime\prime}$ and $h_j^{\prime} = \frac{B_j - \hat{B}}{c} - w \cdot x_j^{\prime\prime}$, then we have $h_i^{\prime} + h_j^{\prime} = \frac{B_i - \hat{B}}{c} - w \cdot x_i^{\prime\prime} + \frac{B_j - \hat{B}}{c} - w \cdot x_j^{\prime\prime}$. Since $x_i^{\prime\prime} < x_j^{\prime}$ and $x_j^{\prime\prime} < x_i^{\prime}$, we have $h_i^{\prime} + h_j^{\prime} \geq \frac{B_i - \hat{B}}{c} - w \cdot x_i^{\prime} + \frac{B_j - \hat{B}}{c} - w \cdot x_j^{\prime} = h_i + h_j$. Due to $x_i^{\prime\prime} < x_j^{\prime}$ and $x_j^{\prime} < x_i^{\prime}$, and $h_i < h_i^{\prime}$ for leaving the same residual energy $\hat{B}$. Moreover, $h_i + h_j \leq h_i^{\prime} + h_j^{\prime} < h_j + h_j^{\prime}$ due to $h_i^{\prime}< h_j$, we have $h_i < h_j^{\prime}$.

After swapping the two UAVs, we have $x_j^{\prime\prime} = x_j^{\prime} - r(h_j) + 2 r(h_i^{\prime}) + r(h_j^{\prime})$. Note that $h_i^{\prime}+h_j^{\prime} \geq h_i+h_j$ implies $h_i^{\prime} + h_j^{\prime} > h_j$. Since $r(h)$ is an increasing function in our interested scope, we have $r(h_i^{\prime}) + r(h_j^{\prime})- r(h_j) > 0$ and further $2r(h_i^{\prime})-r(h_j)+r(h_j^{\prime})>0$. By combining this with the first equation in this paragraph, we have $x_j^{\prime\prime} >  x_j^{\prime}$ and $h_j^{\prime} < h_j$ for leaving the same residual energy $\hat{B}$. Overall, we have $h_i < h_i^{\prime} < h_j$, $h_i < h_j^{\prime} < h_j$ and $h_i + h_j < h_i^{\prime} + h_j^{\prime}$. To show the contradiction, we just need to prove $x_j^{\prime} - r(h_j) + 2 (r(h_i^{\prime}) + r(h_j^{\prime})) > x_j^{\prime} - r(h_j) + 2(r(h_i) + r(h_j))$ for enlarging the total coverage with the same $\hat{B}$. Simply, we only need to prove $r(h_i^{\prime}) + r(h_j^{\prime}) > r(h_i) + r(h_j)$.

We let $p = h_i^{\prime} + h_j^{\prime}$, then $p - h_i > h_j$. We define function $f(\lambda) = r(\lambda) + r(p - \lambda)$ and the derivative is $f^{\prime}(\lambda) = r^{\prime}(\lambda) - r^{\prime}(p - \lambda)$. $f^{\prime}(\lambda) = 0$ implies $\lambda = \frac{p}{2}$. As $r(h)$ is an increasing concave function, $r^{\prime}(h) > 0$ and function $r^{\prime}(h)$ is decreasing. Thus, $f^{\prime}(\lambda) > 0$ for $0 \leq \lambda < \frac{p}{2}$. It follows that $f(\lambda)$ is an increasing function for $0 \leq \lambda < \frac{p}{2}$. Since $p = h_i^{\prime} + h_j^{\prime}$, either $h_i^{\prime}$ or $h_j^{\prime}$ is less than $\frac{p}{2}$. Specifically, if $h_j^{\prime} \leq \frac{p}{2}$, then we have $f(h_j^{\prime}) = r(h_i^{\prime}) + r(h_j^{\prime}) > f(h_i)=r(h_i)+r(p-h_i)$ due to $h_i<h_j^{\prime}$; if $h_i^{\prime} \leq \frac{p}{2}$, then we have $f(h_i^{\prime}) = r(h_i^{\prime}) + r(h_j^{\prime}) >  f(h_i) =r(h_i)+r(p-h_i)$ due to $h_i < h_i^{\prime}$. In both cases, we thus have $r(h_i^{\prime})+r(h_j^{\prime})>r(h_i)+r(h_j)$ due to $p-h_i>h_j$. This completes our proof by contradiction for this final subcase.

\end{proof}

\begin{figure}[t]
    \centering
        \includegraphics[width=0.5\textwidth]{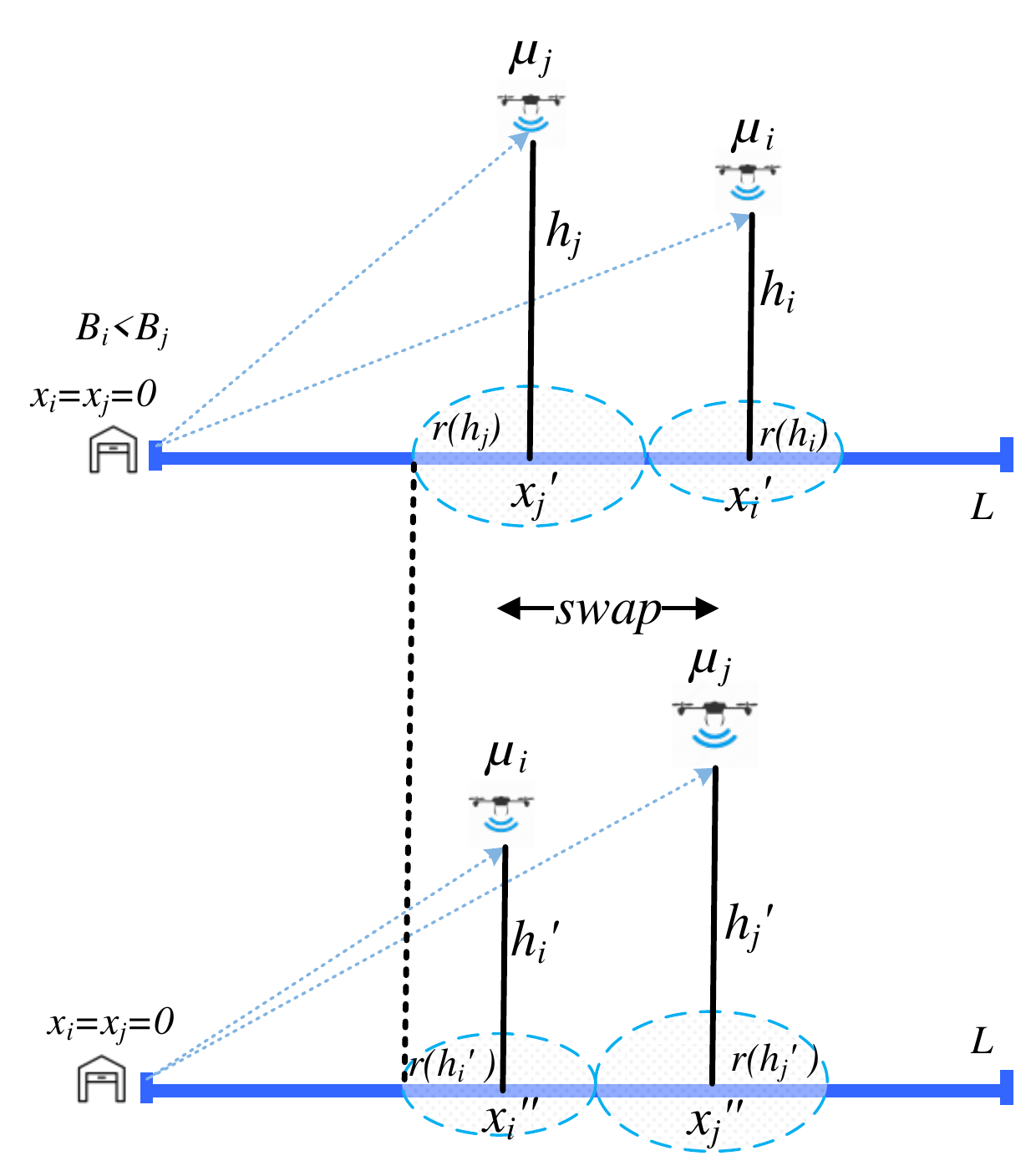}
    \caption{Proof illustration of two UAVs $\mu_i$ and $\mu_j$ with $B_i < B_j$, for showing UAV $\mu_i$ should be deployed closer than $\mu_j$.}
    \label{fig:hold}
\end{figure}

By Proposition~\ref{border}, we can determine the ground destination order of the UAVs according to the order of $B_1 \leq B_2 \leq \cdots \leq B_n$. Based on Propositions~\ref{seamless} and~\ref{border}, we are ready to determine the optimal destination of each UAV in 2D for keeping their leftover energy storages identical. Specifically, due to non-overlapping full coverage of area $[0, L]$, we first have

{
\begin{align}
2 \sum_{1 \leq i \leq n} r(h_i) = L. \label{eq:sum}
\end{align}
}
Recall that, the leftover energy storage of UAV $\mu_i$ is $B_i - c \cdot (w \cdot |x_i^{\prime}| + h_i)$, which should be the same as any other UAV's leftover energy storage, i.e.,
{
\begin{align}
\label{eq:eq}
&~ B_i - c \cdot (w \cdot (2 (r(h_1) + r(h_2) + \ldots + r(h_i)) -r(h_i)) + h_i) = \\ \notag
&~ B_j - c \cdot (w \cdot (2 (r(h_1) + r(h_2) + \ldots + r(h_j)) -r(h_j)) + h_j), ~\text{for}~1 \leq i \neq j \leq n,
\end{align}
}\noindent where any UAV $\mu_i$'s final ground destination is $x_i^{\prime} = 2 (r(h_1) + r(h_2) + \ldots + r(h_{i})) - r(h_i)$.  As we now rewrite $x_i^{\prime}$ as functions of $h_i$'s, we only have $n$ unknowns of $h_i$'s in the $n$ equations in (\ref{eq:sum}) and (\ref{eq:eq}). By solving these equations jointly in a constant time $O(1)$, we obtain the optimal deployment position $(x_i^{\prime}, h_i)$ for each UAV $\mu_i$.

\begin{figure}[t]
\centering
\subfigure[Case 1: UAV $\mu_i$ relocates to the left-edge of NFZ $\delta^l$.]{
\label{fig:subfig:case1}
\includegraphics[width=0.32\textwidth]{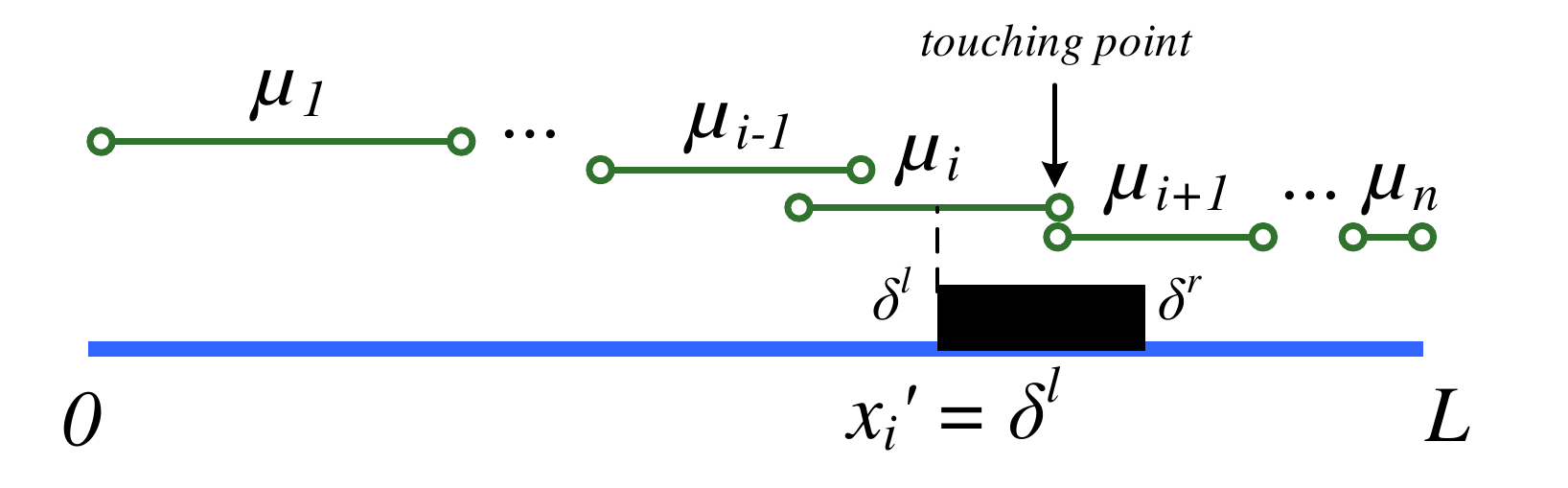}}
\hspace{-0.00in}
\subfigure[Case 2: UAV $\mu_i$ relocates to the right-edge of NFZ $\delta^r$.]{
\label{fig:subfig:case2}
\includegraphics[width=0.3\textwidth]{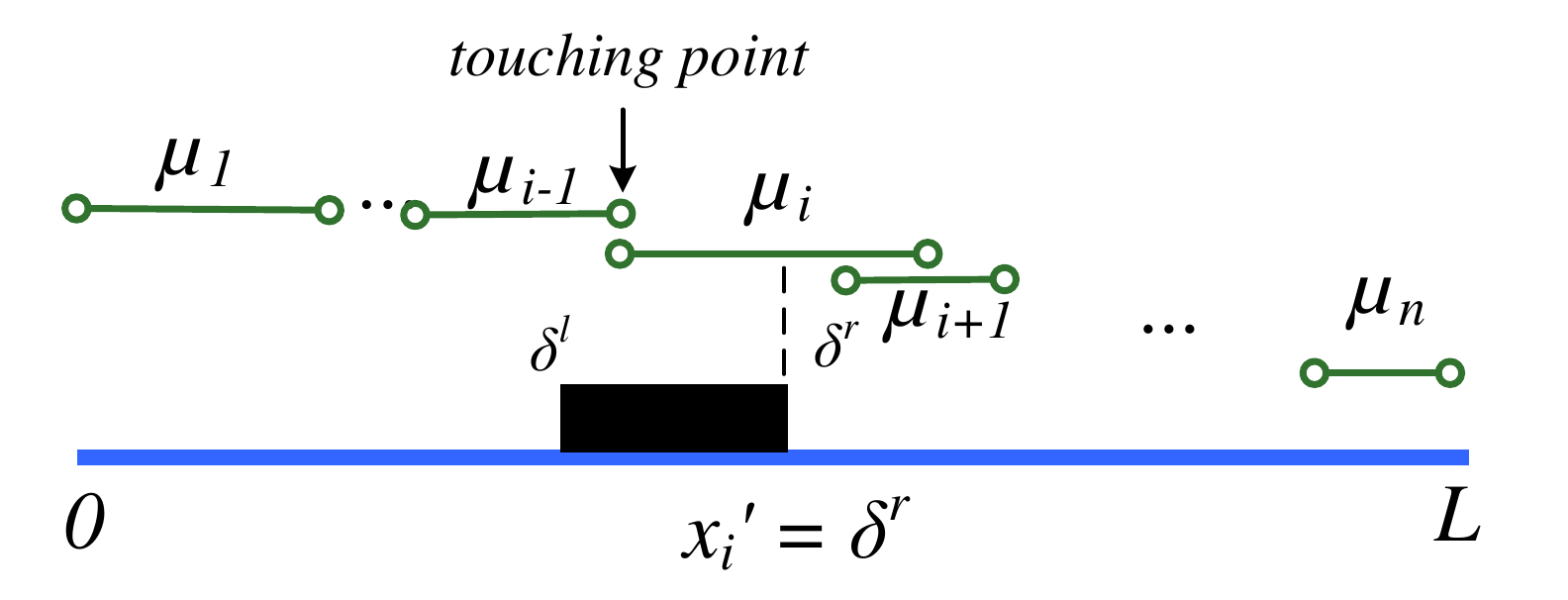}}
\hspace{-0.00in}
\subfigure[Case 3: UAVs $\mu_{i-1}$ and $\mu_{i}$ relocates to $\delta^l$ and $\delta^r$, respectively.]{
\label{fig:subfig:case3}
\includegraphics[width=0.3\textwidth]{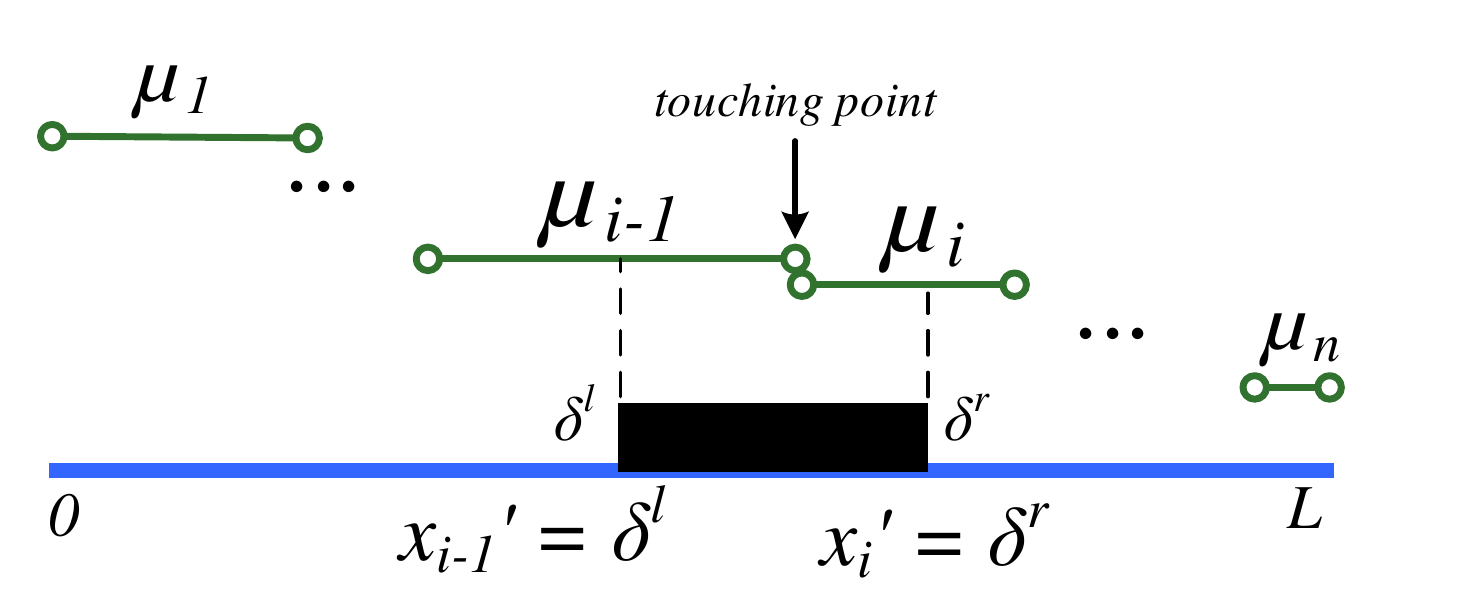}}
\caption{Refinement of the deployment solution by considering the NFZ constraint.} \label{fig:ins}
\end{figure}

\subsection{Incorporation of NFZs for multi-UAV deployment} \label{sec:zone}

As presented in Section~\ref{sec:sameb}, we can compute the maximum minimum leftover energy storage objective $\hat{B}^*$ by solving (\ref{eq:sum}) and (\ref{eq:eq}) without considering NFZs. However, if the destinations of some UAVs fall into NFZ, the solution is not feasible. Without much loss of generality, we assume there is one NFZ, i.e., ($\delta^l$, $\delta^r$).\footnote{If there is more than one NFZ, we can similarly discuss each UAV's possibility to fly into any NFZ and there are just more combinations of Cases 1-3 across NFZs as in this subsection.} In this case, we consider all three possible cases for refining the deployment solution, and we have at most two UAVs relocated to the left- and right-hand-side edges of the NFZ at the optimum.

\emph{Case 1}:~UAV $\mu_i$ is chosen among all to relocate to the left-edge of NFZ (i.e.,$x_i^{\prime} = \delta^l$), as shown in Figure~\ref{fig:subfig:case1}. This UAV is not necessarily the one inside NFZ at the solution to (\ref{eq:eq})-(\ref{eq:sum}). Similar to Proposition~\ref{seamless}, given $B_1 \leq \cdots \leq B_n$, we can show that at the performance bottleneck UAVs $\mu_i, \mu_{i+1}, \cdots, \mu_{n}$ have the same leftover energy storage as $\hat{B}^*$ and they seamlessly cover $[\delta^l -r(h_i), L]$ including NFZ without any overlap in their own coverages. That is,
        {
        \begin{align}
        &~2 \sum_{i \leq j \leq n} r(h_j) = L - \delta^l + r(h_i), \label{eq:cas11} \\
        &~B_j - c \cdot ( w \cdot ( \delta^l + r(h_i) + 2( r(h_{i+1})+ r(h_{i+2}) + \notag \\
        &~\cdots + r(h_{j-1}) + r(h_j)))) + h_j) = \notag \\
        &~B_i - c \cdot (w \cdot \delta^l + h_i) = \hat{B}^{*}, ~~ \forall~j,~ i \leq j \leq n, \label{eq:cas12}
        \end{align}
        }where UAV $\mu_i$ is newly located to ($x_i^{\prime}=\delta^l$, $h_i$) and UAV $\mu_j$ to ($x_j^{\prime}=\delta^l+r(h_i) + 2(r(h_{i+1})+ r(h_{i+2}) + \cdots + r(h_{j-1}) + r(h_j), h_j$). By solving (\ref{eq:cas11})-(\ref{eq:cas12}), we can determine the deployment positions of UAVs $\mu_i, \mu_{i+1}, \cdots, \mu_n$ as well as the objective $\hat{B}^*$. Still we need to check and make sure that the other UAVs $\mu_1, \mu_{2}, \ldots, \mu_{i-1}$ are able to cover $[0, \delta^l - r(h_i)]$ by keeping at least energy $\hat{B}^{*}$ after deployment. Note that they may overlap with UAV $\mu_i$'s wireless coverage without reaching the bottleneck.

\emph{Case 2}:~UAV $\mu_i$ is chosen among all to dispatch to the right-edge of NFZ ($x_i^{\prime} = \delta^r$), as shown in Figure~\ref{fig:subfig:case2}. Similar to Proposition~\ref{seamless}, we can show that at the performance bottleneck, UAVs $\mu_1, \mu_{2}, \ldots, \mu_{i}$ should have the same leftover energy storage as $\hat{B}^*$ and they seamlessly cover $[0, \delta^r]$ without any coverage overlap. That is,
      {\small
        \begin{align}
        &~2 \sum_{1 \leq j \leq i} r(h_j) = \delta^r + r(h_i), \label{eq:cas21} \\
        &~B_j - c \cdot ( w \cdot (2(r(h_1) + r(h_2) + \cdots + r(h_{j-1}))+r(h_{j})) \notag \\
        &~+ h_{j} ) = B_i - c \cdot (w \cdot \delta^l + h_i) = \hat{B}^{*} ~~ \forall~j,~1 \leq j \leq i,\label{eq:cas22}
        \end{align}
      }where UAV $\mu_i$ is located to ($x_i^{\prime}=\delta^r$, $h_i$) in 2D and UAV $\mu_j$ is located to ($x_j^{\prime}= 2(r(h_1) + r(h_2) + \cdots + r(h_{j-1}))+r(h_{j}), h_j$). Still we need to check and make sure that the other UAVs $\mu_i, \mu_{i+1}, \ldots, \mu_{n}$ are able to cover the rest target area $[\delta^r + r(h_i), L]$ by keeping at least energy $\hat{B}^{*}$ after deployment. Note that they may overlap with UAV $\mu_i$'s wireless coverage without reaching the bottleneck.

\emph{Case 3}:~Two neighboring UAVs $\mu_{i-1}$ and $\mu_{i}$ are chosen to dispatch to the both edges of NFZ (i.e., $x_{i-1}^{\prime} = \delta^l, x_i^{\prime} = \delta^r$), as shown in Figure~\ref{fig:subfig:case3}. The NFZ is covered by $\mu_{i-1}$ and $\mu_{i}$ seamlessly and these two UAVs' coverages do not overlap. That is,
      {
        \begin{align}
        &~r(h_{i-1}) + r(h_{i}) =  \delta^r - \delta^l, \label{eq:cas31} \\
        &~B_{i-1} - c \cdot (w \cdot \delta^l + h_{i-1}) = B_{i} - c \cdot (w \cdot \delta^r + h_{i}) = \hat{B}^{*} \label{eq:cas32}
        \end{align}
      }
Moreover, we need to check if UAVs $\mu_1, \mu_{2}, \ldots, \mu_{i-2}$ are able to cover $[0, \delta^l - r(h_{i-1})]$ and if UAVs $\mu_{i+1}, \ldots, \mu_{n}$ are able to cover $[\delta^r + r(h_i), L]$, by keeping at least energy $\hat{B}^*$ after deployment.
In these three cases, we can see that the critical UAV index $i \in \{1, \cdots, n\}$ is still undetermined. We propose to run binary search on UAV set $\{\mu_1, \cdots, \mu_n\}$ to find the optimal $i$, providing the maximum leftover energy storage for the whole UAV network. Note that the binary search has complexity $O(\log n)$ for each UAV in each case. For example, in Case 1, after solving (\ref{eq:cas11})-(\ref{eq:cas12}) for each UAV $\mu_i$, we still need to check the feasibility of the other UAVs $\mu_1, \cdots, \mu_{i-1}$ on the left-hand side to fully cover $[0, \delta^l-r(h_i)]$ in linear running time, resulting in running time $O(n \log n)$ for scanning through all the UAVs. We summarize all the above in Algorithm~\ref{alg:same} and have the following result.

\begin{algorithm}[tbh]
\caption{Deploying the UAVs from the same initial location by considering NFZ}
\begin{algorithmic}[1]

\STATE \textbf{Input:}\\ 
UAV set $\textbf{U} =\{\mu_1, \mu_2, \ldots, \mu_n\}$, \\
NFZ $[\delta^l, \delta^r]$,\\
A continuous line interval $[0, L]$ as target area

\STATE \textbf{Output:}\\ 
$\hat{B}$: the maximum leftover energy storage of the network
\STATE compute UAV final locations by solving equations (\ref{eq:sum}) and (\ref{eq:eq})
\STATE {run binary search to select any critical UAV $\mu_i$ for Cases 1, 2, 3 and compare to choose the maximum leftover storage as $\hat{B}^*$.}
\RETURN $\hat{B}^*$ \\
\end{algorithmic}
\label{alg:same}
\end{algorithm}

\begin{theorem}\label{thry:same}
When dispatching UAVs from the same UAV station under the NFZ constraint, Algorithm~\ref{alg:same} optimally finds the maximum minimum leftover energy storage in $O(n \log n)$ running time.
\end{theorem}

\section{Deploying UAVs from different initial locations} \label{sec:ddif}

In this section, we study the problem when UAVs may be initially located at different locations (e.g., different UAV ground stations or the places that UAVs rest after last task). This is especially the case under emergency when we need a lot more UAVs than those just from the nearest UAV station. Due to the UAV diversity in both initial energy storage and initial location, this problem becomes very difficult. It belongs to the domain of combinatorial optimization and the complexity is generally exponential in the number of UAVs. We aim to design approximation algorithms to maximize the minimum leftover energy storage.

\subsection{When UAVs have identical initial energy storage} \label{sec:different}

In this subsection, we first study the special case when all UAVs have identical initial energy storage $B$. Without loss of generality, we assume the UAVs are sorted according to the increasing order of their initial ground locations, i.e., $x_1 \leq x_2 \cdots \leq x_n$. We first prove that all UAVs' destinations after deployment should follow the same order as the initial locations.

\begin{proposition}\label{order}
Given that the initial locations of the $n$ UAVs follow the order $x_1 \leq x_2 \leq \cdots \leq x_n$, the final ground destinations of them preserve the same order $x_1^{\prime} \leq x_2^{\prime} \leq \cdots \leq x_n^{\prime}$ in the optimal solution of deployment.
\end{proposition}
\begin{proof}
Consider any two neighboring UAVs $\mu_i$, $\mu_j$ with initial locations $x_i \leq x_j$ on the ground. We prove by contradiction here. After deployment, UAVs $\mu_i$ and $\mu_j$ cover a continuous portion of the line interval $l$, and suppose $x_i^{\prime} > x_j^{\prime}$ in an optimal solution, as illustrated on the left-hand side of Figure~\ref{fig:order}. We obtain their travel distances $d_i = w \cdot |x_i-x_i^{\prime}| + h_i$ and $d_j = w \cdot |x_j- x_j^{\prime}| + h_j$. If we swap the locations of $\mu_i$ and $\mu_j$ without changing their altitudes as illustrated on the right hand side of Figure~\ref{fig:order}, the new travel distances $d_i$ and $d_j$ are smaller than before and they save more energy by keeping the same total coverage $l$. Then the UAV network's leftover energy increases and the previous allocation of $x_i^{\prime}\geq x_j^{\prime}$ is actually not optimal. This is the contradiction and we require $x_i^{\prime} \leq x_j^{\prime}$ given $x_i \leq x_j$.
\end{proof}

\begin{figure}[t]
    \centering
        \includegraphics[width=0.7\textwidth]{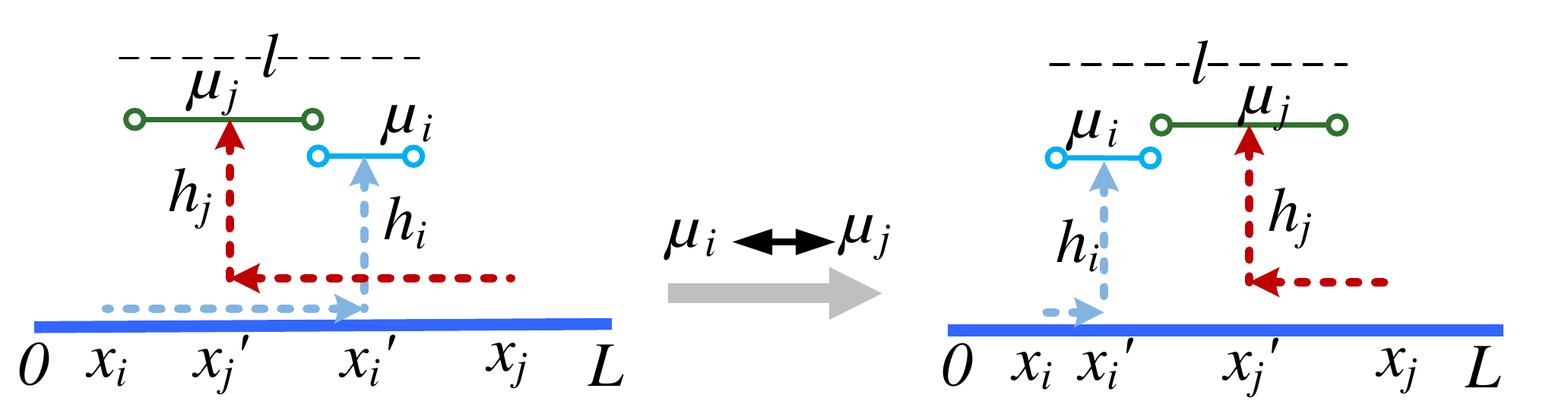}
    \caption{Illustration of initial order preserving of UAVs.}
    \label{fig:order}
\end{figure}

Proposition~\ref{order} greatly simplifies the algorithm design of deploying UAVs by fixing their location order after deployment, and we do not need to consider many combinations of UAVs' possible location orders in the combinatorial optimization. It also holds after incorporating NFZs and we next design the algorithm in two stages. First, we introduce the feasibility checking problem and design Algorithm~\ref{alg:decalg} to determine whether we can find a deployment scheme for any given leftover energy storage $\hat{B}$. Then, we run binary search over all these feasible energy storages to find the optimum in Algorithm~\ref{alg:fptas}.
\subsubsection{Feasibility checking problem}\label{sec:fea1}

\begin{algorithm}[t]
\caption{Feasibility checking algorithm to keep leftover energy storage $\hat{B}$}
\begin{algorithmic}[1]

\STATE \textbf{Input:}\\ 
$\textbf{U} =\{\mu_1, \mu_2, \ldots, \mu_n\}$ \\
$\hat{B}$: a given amount of leftover energy storage for each UAV

\STATE \textbf{Output:}\\ 
Feasibility result of $\hat{B}$

\STATE Compute $a_{i}$ in equation (\ref{eq:a}) and $b_{i}$ in equation (\ref{eq:b})
\STATE $\overline{L} = 0$;

\FOR { $i = 1$ to $n$ }
\IF {$\overline{L} \in [a_i, b_i]$}
	\STATE {$x_i^{\prime} \gets \min\{ \overline{L} + r(h_i), b_i(h_i)-r(h_i)\}$}
     \IF {$\delta^l <  x_i^{\prime} < \delta^r$}
        \STATE {$x_i^{\prime} \gets \delta^l$}
        \STATE {$\overline{L} \gets \max\{\overline{L}, \delta^l + r(\frac{B - \hat{B}}{c} - w|x_i^{\prime}- x_i|)\}$}
    \ENDIF
	\STATE {$\overline{L} \gets x_i^{\prime}+r(h_i)$}
\ENDIF
\ENDFOR

\IF {$\overline{L} < L$}
	\RETURN $\hat{B} \ is \ not feasible$ ($\hat{B} > \hat{B}^*$)\\
\ELSE
	\RETURN $\hat{B} \ is \ feasible$ ($\hat{B} \leq \hat{B}^*$) and update final positions $(x_i^\prime,h_i)$'s \\
\ENDIF

\end{algorithmic}
\label{alg:decalg}
\end{algorithm}

Given any amount of leftover energy storage $\hat{B} \in (0, B)$, we want to determine whether the energy budget $B - \hat{B}$ is feasible to support all the UAVs to reach a full coverage of $[0, L]$ and avoid sitting into NFZs. Let $\hat{B}^*$ denotes the maximum minimum leftover storage, we next design the feasibility checking algorithm to determine whether $\hat{B} > \hat{B}^*$ (infeasible for the UAVs to cover the whole target area) or $\hat{B} \leq \hat{B}^*$ (feasible). Note that $\hat{B}$ is unknown yet and will be determined in next subsection. Our feasibility checking algorithm also works for the case with NFZs and its complexity order does not increase in the number of NFZs.

For UAV $\mu_i$, we respectively denote $a_i$ as the leftmost ground point and $b_i$ as the rightmost point on $L$ that can be covered by this UAV with $\hat{B}$ leftover storage and altitude $h_i$. Note that $a_i<x_i<b_i$. To cover $a_i$, UAV $\mu_i$ travels horizontally $x_i-a_i-r(h_i)$ distance, and it travels horizontally $b_i-r(h_i)-x_i $ distance to cover $b_i$. Then we have
{
\setlength{\abovedisplayskip}{3pt}
\setlength{\belowdisplayskip}{1pt}
\begin{align}
\label{eq:a}
a_i(h_i)= (\frac{\hat{B} - B}{c \cdot w}) + x_i + \frac{h_i}{w}- r(h_i),
\end{align}
}
{
\setlength{\abovedisplayskip}{3pt}
\setlength{\belowdisplayskip}{3pt}
\begin{align}\label{eq:b}
b_i(h_i)= (\frac{B - \hat{B}}{c \cdot w}) + x_i - \frac{h_i}{w} + r(h_i),
\end{align}
}

\noindent both of which are functions of altitude $h_i$. Without loss of generality, we sequentially deploy UAVs to cover the target area from the left to right hand side, and we denote the currently covered interval as $[0, \bar{L}]$. 

In Algorithm~\ref{alg:decalg}, we first respectively compute $a_i$ and $b_i$ in equations (\ref{eq:a}) and (\ref{eq:b}) in line 3, and then deploy the UAVs one by one (in line 5) according to their initial locations' order to cover from the left endpoint of target interval $[0, L]$ as in Proposition~\ref{order}. As $x_1 \leq x_2 \leq \cdots \leq x_n$, we start with UAV $\mu_1$ and end up with UAV $\mu_n$. Specifically, given our currently covered interval $[0, \overline{L}]$, we check whether UAV $\mu_i$ can extend $\bar{L}$ within its energy budget $B-\hat{B}$ (i.e., $\overline{L} \in [a_i, b_i]$ in line 6). If so, we will deploy UAV $\mu_i$'s ground destination to $x_i^{\prime} = \min(\overline{L}+r_i(h_i), b_i(h_i)- r_i(h_i))$ in line 7 and cover from point $\bar{L}$. Note that we do not use a UAV if $\bar{L}$ is not within $[a_i, b_i]$. If the computed $x_i^{\prime}$ falls into some NFZ $(\delta^l, \delta^r)$, we should deploy UAV $\mu_i$ to $\delta^l$ on the ground and the corresponding service altitude is $\frac{B-\hat{B}}{c} - w \cdot |x_i - \delta^l|$ in the sky, as shown in lines 6-13. If the outcome of the algorithm shows $\hat{B}$ is feasible i.e., ($\hat{B} \leq \hat{B}^*$) in line 18, our algorithm will further return the UAVs' final locations $(x_i^{\prime}, h_i)$'s in 2D. Otherwise, it returns that $\hat{B}$ is infeasible in line 16. Overall, Algorithm~\ref{alg:decalg} solves the feasibility checking problem given a particular leftover energy storage $\hat{B}$ in linear running time.

\begin{algorithm}[t]
\caption{Approximation for multi-UAV deployment from different initial locations}
\begin{algorithmic}[1]

\STATE \textbf{Input:}\\ 
$\Lambda = \{\epsilon \hat{B}_l, 2 \epsilon \hat{B}_l, \cdots,  \lceil{\frac{\hat{B}_u}{\epsilon \cdot \hat{B}_l}}\rceil \epsilon \hat{B}_l\}$ \\

\STATE \textbf{Output:}\\ 
$\Lambda(ind)$: $ind$ is the selected index

\STATE $low \gets 1$ and $high \gets \lceil{\frac{\hat{B}_u}{\epsilon \cdot \hat{B}_l}}\rceil$
\WHILE{$low <= high$}
       \STATE {$mid \gets \lfloor{(low + high)/2}\rfloor$}
       \STATE {feasibility checking on $\Lambda(mid)$ by Algorithm~\ref{alg:decalg}}
       \IF {$\Lambda(mid)$ is feasible}
            \STATE {$high \gets mid$}
       \ELSE
            \STATE {$low \gets mid$}
       \ENDIF
        \IF {$low = high-1$}
            \STATE {$ind \gets high$}
            \STATE {break}
        \ENDIF
\ENDWHILE
\RETURN $\Lambda(ind)$
\end{algorithmic}
\label{alg:fptas}
\end{algorithm}

\subsubsection{Binary search over all feasible energy storages}\label{sec:bs}
With the help of Algorithm~\ref{alg:decalg}, we can verify whether a given leftover energy storage $\hat{B}$ is feasible or not. The maximum leftover storage among all the feasible ones is actually the optimum. Here, we apply binary search to find the maximum leftover storage upon determining the search scope of $\hat{B}$ and step size.

We first determine the search scope by computing its upper and lower bounds. By using the $n$ UAVs to cover the line interval of length $L$, we must have at least one UAV with coverage radius $r \geq \frac{L}{2n}$. For each UAV $\mu_i$, the minimum moving distance is its service altitude $h_i$ without any ground movement. Suppose UAV $\mu_i$ is the one with coverage radius $r(h_i) \geq \frac{L}{2n}$ or equivalently altitude $h_i \geq r^{-1}(\frac{L}{2n})$, the upper bound of minimum leftover energy storage among all the UAVs is $\hat{B}_u = B - c \cdot r^{-1}(\frac{L}{2n})$.

We next determine the lower bound. We know that the longest horizontal travel distance of UAV $\mu_i$ is $\max\{|L - x_i|, |x_i|\}$. Without flying any distance horizontally but vertically to the sky, the maximum service altitude is $r^{-1}(\max\{|x_i|, |L-x_i|\})$ to cover $[0, L]$. Then we have the lower bound $\hat{B_l}=B- c(w \max\{|x_i|, |L-x_i|\} + r^{-1}(\max\{|x_i|, |L-x_i|\}))$.\footnote{In case that this formula returns a negative lower bound, we can replace it by the smallest possible positive energy storage.}

Now we are ready to present the approximation algorithm in Algorithm~\ref{alg:fptas} by combining both binary search and feasibility checking in Algorithm~\ref{alg:decalg}. We denote the relative error as $\epsilon$, and accordingly set the search accuracy as $\epsilon \hat{B}_l$ in line 1 of Algorithm~\ref{alg:fptas}. Algorithm~\ref{alg:fptas} starts with $\epsilon \hat{B}_l$ in line 3 and stops once turning from feasible leftover energy storage $\hat{B}^{\prime}$ to infeasible $\hat{B}^{\prime\prime}$ in lines 12-15 as illustrated in illustrated in Figure~\ref{fig:binary_search}. Then, the final $\hat{B}^{\prime}$ in line 17 is our searched optimum and is denoted as $\hat{B}^*$.

\begin{figure}[t]
    \centering
        \includegraphics[width=0.7\textwidth]{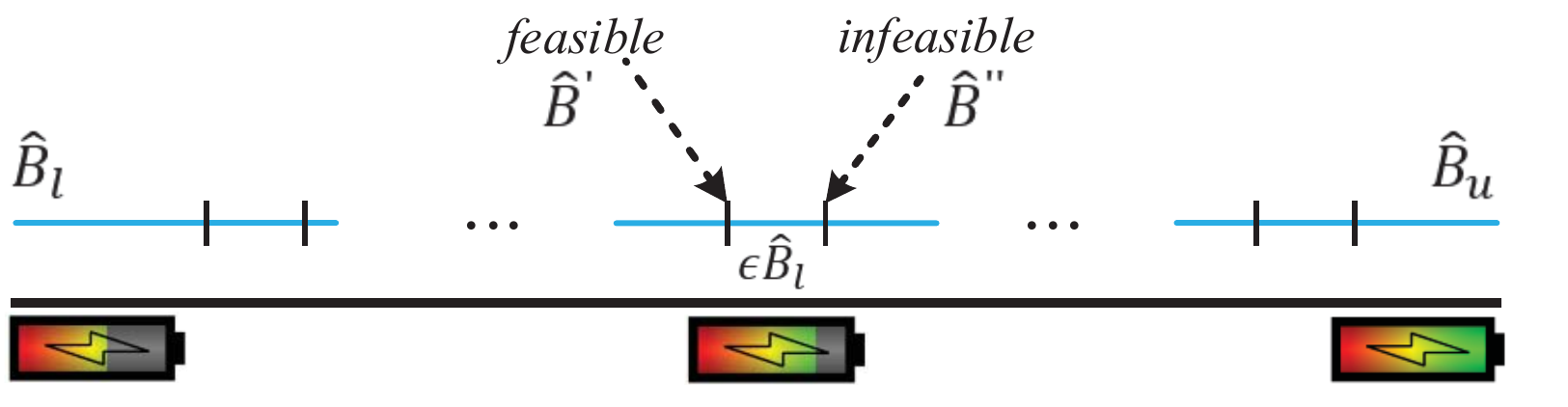}
    \caption{Binary search on $[\hat{B}_l ,\hat{B}_u]$ with accuracy level of $\epsilon \cdot \hat{B}_l$.}
    \label{fig:binary_search}
\end{figure}
\begin{theorem}\label{thrm:FPTAS}
Let $\hat{B}^*$ be the optimal leftover energy storage in problem (\ref{general}). Algorithm~\ref{alg:fptas} presents an $(1 - \epsilon)$-approximation algorithm of computational complexity $O(n \log \frac{1}{\epsilon})$.
\end{theorem}
\begin{proof}
The leftover energy storage of a given instance has an upper bounded of $\hat{B}_u$ and a lower bound of $\hat{B}_l$. Obviously, $\hat{B}_l \leq \hat{B}^* \leq \hat{B}_u$. Choosing a small constant $\epsilon > 0$, we divide each $\hat{B}_l$ into $\frac{1}{\epsilon}$ sub-intervals. Each interval has length $\epsilon \cdot \hat{B}_l$, where $\epsilon \cdot \hat{B}_l \leq \epsilon \cdot \hat{B}^*$. We divide $\hat{B}_u$ by $\epsilon \cdot \hat{B}_l$ into $\lceil{\frac{\hat{B}_u}{\epsilon \cdot \hat{B}_l}}\rceil$ sub-intervals as set $\Lambda$ as in line 1 of Algorithm~\ref{alg:fptas}.

Then, each operation in line 6 of binary search will reduce set $\Lambda$ by applying Algorithm~\ref{alg:decalg} on a given $\hat{B}$. It finally terminates with leftover energy storages $\hat{B}^{\prime}$ and $\hat{B}^{\prime\prime}$, as illustrated in Figure~\ref{fig:binary_search}, where $\hat{B}^{\prime} < \hat{B}^*$ and $\hat{B}^{\prime\prime} = \hat{B}^{\prime} + \epsilon \cdot \hat{B}_l > \hat{B}^*$. The outcome $\hat{B}^{\prime}$ is our searched optimum. We can obtain that $\hat{B}^{\prime} = \hat{B}^{\prime\prime} - \epsilon \cdot \hat{B}_l \geq \hat{B}^* - \epsilon \cdot \hat{B}_l \geq (1 - \epsilon) \hat{B}^*$. That is, we have $(1 - \epsilon) \hat{B}^* \leq \hat{B}  < \hat{B}^*$. Therefore, we obtain the minimum leftover storage with approximation ratio $1-\epsilon$ as compared to the global optimum. Regarding the computational complexity, note that our feasibility checking algorithm runs in $O(n)$ time. Moreover, we have $O(\lceil\frac{\hat{B}_u}{\epsilon \cdot \hat{B}_l}\rceil)$ candidates of leftover energy storage for binary search with logarithmic running time, in which $\frac{\hat{B}_u}{\hat{B}_l}$ is constant and independent of $n$. Overall, this algorithm runs in $O(n \log \frac{1}{\epsilon})$ time.
\end{proof}

\subsection{Deploying UAVs with different initial energy storages}\label{sec:diffb}

When UAVs are required to undertake emergency tasks, some of them may not be fully charged yet and they have different initial energy storages in general. To be specific, each UAV $\mu_i$ has an initial storage $B_i$ upon deployment. In this subsection, we further consider that UAVs may have different initial energy storages under the constraint of NFZs. The UAVs' two-dimensions of heterogeneity (in both initial locations and initial energy storages), joint ground-sky movements and NFZ consideration results in high complexity for UAVs' energy-saving deployment algorithm design. In the following, we show such a general UAV deployment problem is NP-hard, and propose a heuristic algorithm accordingly.


Given any value of $\hat{B}$ as the minimum leftover energy storage among the UAVs after deployment, we want to determine whether UAVs can be moved to reach a full coverage of $[0, L]$ under the NFZs constraint. We call it feasibility problem as in previous Section~\ref{sec:fea1}, which is now proved to be NP-complete by reducing from the well-known {\em Partition problem} \cite{garey2002computers}. In this reduction, the basic idea is to consider two NFZs around the middle point of the target line interval, and partition the UAVs with different coverage energy storages into two subsets, each of which seamlessly covers a half of the line interval without any overlaps.

\begin{theorem}\label{thry:np}
When UAVs have different initial energy storages and are deployed from different initial locations under the NFZ consideration, our energy-saving UAV swarm deployment problem in (\ref{general}) is NP-hard.
\end{theorem}
\begin{proof}
See Appendix~\ref{ap:npc}.
\end{proof}

\begin{algorithm}[t]
\caption{$\kappa$-heuristic for the UAVs with different initial locations and energy storages}
\begin{algorithmic}[1]

\STATE \textbf{Input:}\\ 
$\textbf{U} =\{\mu_1, \mu_2, \ldots, \mu_n\}$,\\
$\kappa$: order changing degree

\STATE \textbf{Output:}\\ 
$\hat{B}$: maximum minimum energy storage

\STATE select $\kappa$ UAVs and enumerate all possible permutations, which produces $N = C_n^{\kappa} \cdot \kappa !$ UAV ordering sequences ($\{s_1, s_2, \ldots, s_{N}\}$) \\
\FOR {$i = 1$ to $N$ }
    \STATE {run Algorithm~\ref{alg:fptas} on $s_i$ to obtain $\hat{B}_i$}
\ENDFOR
\STATE $\hat{B} \gets \max\{\hat{B}_i\}$ \\

\RETURN $\hat{B}$
\end{algorithmic}
\label{alg:Heuristic}
\end{algorithm}


The problem (\ref{general}) is difficult to solve due to the UAVs' distinct initial locations and initial energy storages, which result in exponential number of sequences of UAVs for searching. Yet if the location order after deployment is determined (e.g., with initial location order preserving as in Algorithm~\ref{alg:fptas} under the same initial energy storage), we can compute the maximum leftover energy storage fast under the given location order. Further, if we enumerate the location order among UAVs as many as possible, then we can reach the optimality without any efficiency loss. To balance complexity and efficiency, we accordingly propose Algorithm~\ref{alg:Heuristic} with $\kappa$ as the order changing degree after deployment. Specifically, we select $\kappa$ UAVs and enumerate all possible permutations for their final ground locations' ordering. This produces $N = C_n^{\kappa} \cdot \kappa !$ sequences of UAVs'ordering, each of which runs Algorithm~\ref{alg:fptas} separately. If we care more about the efficiency rather than the complexity, we can choose a larger $\kappa$ and when $\kappa = n$, Algorithm~\ref{alg:Heuristic}'s performance is arbitrarily approach the global optimum as in Theorem~\ref{thrm:FPTAS}.

\begin{figure}[t]
    \centering
        \includegraphics[width=0.7\textwidth]{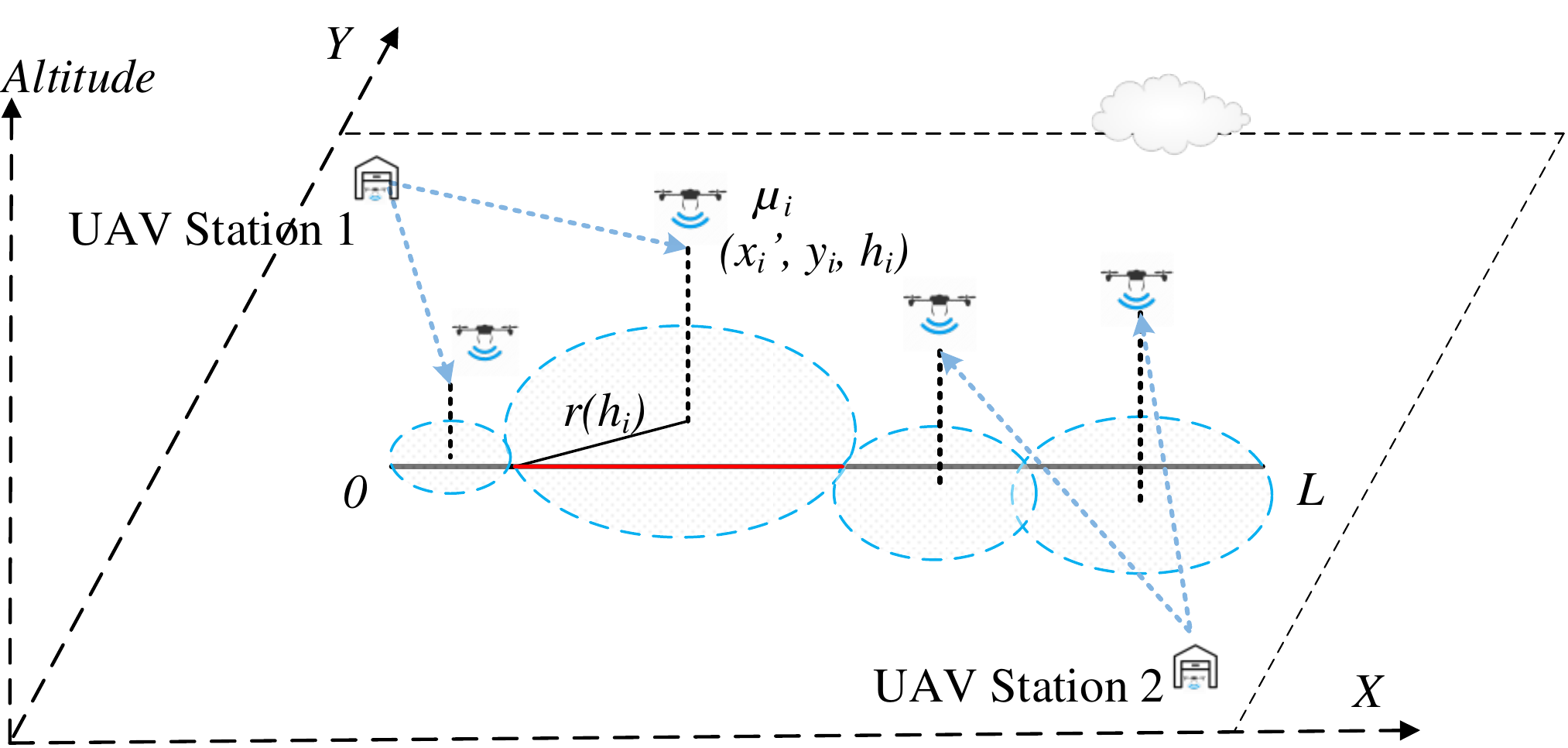}
    \caption{An illustrative example of deploying $n=4$ UAVs that initially located in two UAV stations on the two ends of the 2D ground plane. Here, each UAV station is not necessarily in the same line along the target interval $[0, L]$.}
    \label{fig:covermodel3Dswap}
\end{figure}

\section{Extension to UAVs'energy-saving deployment in a 3D space} \label{sec:3D}

Recall in all the previous sections, we model the target service area as a line interval $[0, L]$ and assume all the UAVs are initially located in the same 1D linear domain. Then we deploy each UAV in 2D by considering the 1D ground and 1D sky movements. In this section, we consider that UAVs' initial locations may not be along the target interval and need to determine UAVs' ground movements in 2D as well as 1D movement to the Sky. Without much loss of generality, we consider two UAV stations on the left and right corners of the ground plane with locations $(x_l, y_l, 0)$ and $(x_r, y_r, 0)$ satisfying $x_l<x_r$, and also place target interval $[0, L]$ in the x-axis in the 2D ground plane, as illustrated in Figure \ref{fig:covermodel3Dswap}. That is, target interval ranges from point $(0, 0 ,0)$ to $(L, 0, 0)$ in the 3D space. Then each UAV $\mu_i$ is now placed finally to point $(x_i^{\prime}, y_i^{\prime}, h_i)$ in 3D space. We want to extend the previously proposed UAV swarm deployment algorithms with theoretical guarantee to 3D. There are $n$ UAVs in total with $|U_l+U_r| = n$ including UAVs $\mu_1 - \mu_{|U_l|}$ initially located in the left UAV station (belonging to set  $U_l$) and UAVs $\mu_{|U_l|+1} - \mu_n$ initially located in the right UAV station (belonging to set $U_r$). We reorder the UAVs from $U_l$ such that $B_1 \leq \cdots \leq B_{|U_l|}$ in increasing energy storage order, while reorder the UAVs from $U_r$ such that $B_{|U_l|+1} \geq \cdots \geq B_n$ in decreasing order.

\begin{lemma}\label{biorder}
At the optimum, the UAVs from the left- and right-hand side UAV stations will not cross each other. That is, $x_1^{\prime} \leq \cdots \leq x_n^{\prime}$ in the $x$-domain in the 3D space. And finally the UAVs will keep the same leftover energy after the deployment.
\end{lemma}
\begin{proof}
First of all, for UAV subset $U_l$ or $U_r$ in the left or right-hand side UAV station, we can prove Proposition~\ref{border} still holds here. That is, we should deploy all the UAVs in $U_l$ from left to right hand side of the target interval according to increasing order of their initial energy storages. Similarly, we deploy the UAVs in $U_r$ from right to left hand side of the target interval according to increasing order of by their initial energy storages. At the optimum, all the UAVs' leftover energy is the same to maximize the minimum leftover energy storage.

Now, we are ready to prove that any two UAVs will not fly to cross each other by contradiction. Note that the UAVs from the same UAV station will not fly across each other, and we will only look at the UAVs' possible crossing from different stations. Suppose that there are two UAVs $\mu_i$ and $\mu_j$ originally in $U_l$ and $U_r$ (i.e., $x_i=x_l$ and $x_j=x_r$), respectively, and they will cross each other with $x_i^{\prime} \geq x_j^{\prime}$ and become neighbors after deployment. They cover a continuous line interval $[x_j^{\prime}-\sqrt{r(h_j)^2- {y_j^{\prime}}^2}, x_i^{\prime}+\sqrt{r(h_i)^2-{y_i^{\prime}}^2}]$ as part of the target interval $[0, L]$, and they have the same leftover energy ($B_i - d_i = B_j - d_j = \hat{B}^*$) in the optimal solution. Next we swap UAVs $\mu_i$ and $\mu_j$ to show a better solution is actually achieved. If we swap the $x$-domain ground location order of $\mu_i$ and $\mu_j$ without changing any UAV's coverage radius or altitude, we can see that $d_i^{\prime} = w \cdot \sqrt{(x_l-x_i^{\prime})^2 + (y_l-y_i^{\prime})^2}+h_i$ and $d_j^{\prime} =  w \cdot \sqrt{(x_r-x_j^{\prime})^2+ (y_r-y_j^{\prime})^2}+h_j$ will decrease, i.e., $d_i^{\prime} < d_i$ and $d_j^{\prime}< d_j$. Specifically, we move $\mu_i$ to cover from $x_j^{\prime} - \sqrt{r(h_j)^2-{y_j^\prime}^2}$ and move $\mu_j$ to cover until $x_i^{\prime} + \sqrt{r(h_i)^2-{y_i^\prime}^2}$. We can see that the original covered range $[x_i^{\prime} - \sqrt{r(h_i)^2- {y_i^{\prime}}^2}, x_j^{\prime} -\sqrt{r(h_j)^2+ {y_j^{\prime}}^2}]$ is still covered by UAVs $\mu_i$ and $\mu_j$, and $d_i^{\prime} < d_i$ and $d_j^{\prime}< d_j$. Given $B_i - c d_i^{\prime}  > \hat{B}^*$ and $B_j - c d_j^{\prime} > \hat{B}^*$, a better solution is achieved to prolong the UAV network lifetime.
\end{proof}

We decide the cooperation among UAVs by separating the UAVs into two subsets $U_l$ and $U_r$, which cover left and right hand parts of $[0, L]$, respectively. Lemma \ref{biorder} helps fix the final location order in $x$-domain of all UAVs and based on the given ordering sequence, we next introduce the feasibility checking problem and binary search similar to Section~\ref{sec:different}.

Before extending Algorithm~\ref{alg:fptas} to our new deployment problem in 3D  given any leftover energy storage $\hat{B}$ and altitude $h_i$ for UAV $\mu_i$, we first determine the leftmost point $a_i$ and the rightmost point $b_i$ on $L$ that can be covered by UAV $\mu_i$ as follows:
{
\setlength{\abovedisplayskip}{3pt}
\setlength{\belowdisplayskip}{1pt}
\begin{align}
\label{eq:a1}
a_i(h_i)= x_l - \sqrt{(\frac{B_i - \hat{B}}{c \cdot w} - \frac{h_i}{w})^2 - (y_i^{\prime} - y_l)^2} - r(h_i),
\end{align}
}
{
\setlength{\abovedisplayskip}{3pt}
\setlength{\belowdisplayskip}{3pt}
\begin{align}\label{eq:b1}
b_i(h_i)= x_r + \sqrt{(\frac{B_i - \hat{B}}{c \cdot w} - \frac{h_i}{w})^2 - (y_i^{\prime} - y_r)^2} + r(h_i).
\end{align}
}

Here we deploy the UAVs in $U_l$ to cover the target area from left to right hand side, and we denote the currently covered interval as $[0, \overline{L}_l]$, while we deploy the UAVs in $U_r$ to cover the target area from right to left hand side, and the currently covered interval is $[\overline{L}_r, L]$.

\begin{algorithm}[tbh]
\caption{Approximation algorithm for deploying UAVs in 3D}
\begin{algorithmic}[1]


\STATE \textbf{Input:}\\ 
$\Lambda = \{\epsilon \hat{B}_l, 2 \epsilon \hat{B}_l, \cdots,  \lceil{\frac{\hat{B}_u}{\epsilon \cdot \hat{B}_l}}\rceil \epsilon \hat{B}_l\}$

\STATE \textbf{Output:}\\ 
$\Lambda(ind)$: $ind$ is the selected index

    \STATE $low \gets 1$ and $high \gets \lceil{\frac{\hat{B}_u}{\epsilon \cdot \hat{B}_l}}\rceil$
    \WHILE{$low <= high$}
           \STATE {$mid \gets \lfloor{(low + high)/2}\rfloor$}
           \STATE {$\hat{B} \gets \Lambda(mid)$}
           \STATE Compute $a_{i}$ for $U_r$ in equation (\ref{eq:a1}) and $b_{i}$ for $U_l$ in equation (\ref{eq:b1})
            \STATE $\overline{L}_l = 0, \overline{L}_r = L$;

            \FOR { $i = 1$ to $ |U_l|$ }
            \IF {$\overline{L}_l \in [b_i-\sqrt{r(h_i)^2-{y_i^{\prime}}^2}, b_i]$}
    	       \STATE {$x_i^{\prime} \gets \overline{L}_l + \sqrt{r(h_i)^2-{y_i^{\prime}}^2}$}
    	\STATE {$\overline{L}_l \gets x_i^{\prime}+\sqrt{r(h_i)^2-{y_i^{\prime}}^2}$}
        \ENDIF
        \ENDFOR

        \FOR { $i =  |U_l|+1$ to $n$ }
        \IF {$\overline{L}_r \in [a_i, a_i+\sqrt{r(h_i)^2-{y_i^{\prime}}^2}]$}
        	\STATE {$x_i^{\prime} \gets \overline{L}_r - \sqrt{r(h_i)^2-{y_i^{\prime}}^2}$}
        	\STATE {$\overline{L}_r \gets x_i^{\prime}-r(h_i)$}
        \ENDIF
        \ENDFOR

       \IF {$\overline{L}_l \geq \overline{L}_r$}
            \STATE {$high \gets mid$}
       \ELSE
            \STATE {$low \gets mid$}
       \ENDIF
        \IF {$low = high-1$}
            \STATE {$ind \gets high$}
            \STATE {break}
        \ENDIF
\ENDWHILE
\RETURN $\Lambda(ind)$

\end{algorithmic}
\label{alg:FPTAS2}
\end{algorithm}

Next, we determine the scope of binary search. The leftover energy $\hat{B}$ among all the UAVs is upper bounded by $\hat{B}_u = \max{B_i}$. We choose a smallest possible energy storage $\hat{B}_l$ as the lower bound of $\hat{B}$. Obviously, $\hat{B}_l \leq \hat{B}^* \leq \hat{B}_u$. Similar to Algorithm~\ref{alg:fptas}, we choose a small positive constant $\epsilon > 0$ and divide each $\hat{B}_l$ into $\frac{1}{\epsilon}$ sub-intervals. Each interval has length $\epsilon \cdot \hat{B}_l$, where $\epsilon \cdot \hat{B}_l \leq \epsilon \cdot \hat{B}^*$. We divide $\hat{B}_u$ by $\epsilon \cdot \hat{B}_l$ to partition into $\lceil{\frac{\hat{B}_u}{\epsilon \cdot \hat{B}_l}}\rceil$ sub-intervals in set $\Lambda$ as the input of Algorithm~\ref{alg:FPTAS2}.

In Algorithm~\ref{alg:FPTAS2}, we first compute $a_i$ and $b_i$ in line 7 according to equations (\ref{eq:a1}) and (\ref{eq:b1}), then deploy the UAVs one by one (line 9 for the UAVs in $U_l$ and line 15 for those in $U_r$) according to their initial energy storages' order to cover the target interval $[0, L]$. Specifically, considering the UAVs in $U_l$, given our currently covered interval $[0, \overline{L}_l]$, iteration $i$ starts with checking whether UAV $\mu_i$ inside can extend the current covered area. That is, we need to check if $\mu_i$ can seamlessly cover from the point $\overline{L}_l$ (i.e., $b_i(h_i) - r(h_i) \leq \overline{L}_l \leq b_i(h_i)$ in line 10). If so, we will deploy UAV $\mu_i$ to $x_i^{\prime} = \overline{L}_l + r(h_i)$ and update $\overline{L}_l$  in lines 11-12. Similar procedure with reverse direction in lines 15-20 follows for those UAVs in set $U_r$.

If $\overline{L}_l \geq \overline{L}_r$ to fully cover the target interval $[0,L]$, the feasibility checking returns that the given $\hat{B}$ is feasible (i.e., $\hat{B} \leq \hat{B}^*$) in line 22, our algorithm will skip this choice $\hat{B}$ and keep running binary search for larger leftover energy solution (if feasible). Otherwise, it returns that $\hat{B}$ is infeasible in line 24 and our algorithm will search for smaller feasible candidates. Finally, our algorithm will return the optimal $\hat{B}^*$ in the search scope. Similar to Theorem~\ref{thrm:FPTAS}, we have the following corollary.

\begin{corollary} \label{coro:FPTAS}
Let $\hat{B}^*$ be the optimal leftover energy storage of our energy-saving UAV deployment problem by dispatching the UAVs swarm from two UAV stations, given any allowable error $\epsilon >0$, Algorithm~\ref{alg:FPTAS2} achieves approximation ratio $(1 - \epsilon)$ and runs in $O(n \log \frac{1}{\epsilon})$ time.
\end{corollary}

\begin{figure}[t]
\centering
\subfigure[Same initial energy storage and without the consideration of NFZ.]{
\label{fig:subfig:ins1}
\includegraphics[width=0.7\textwidth]{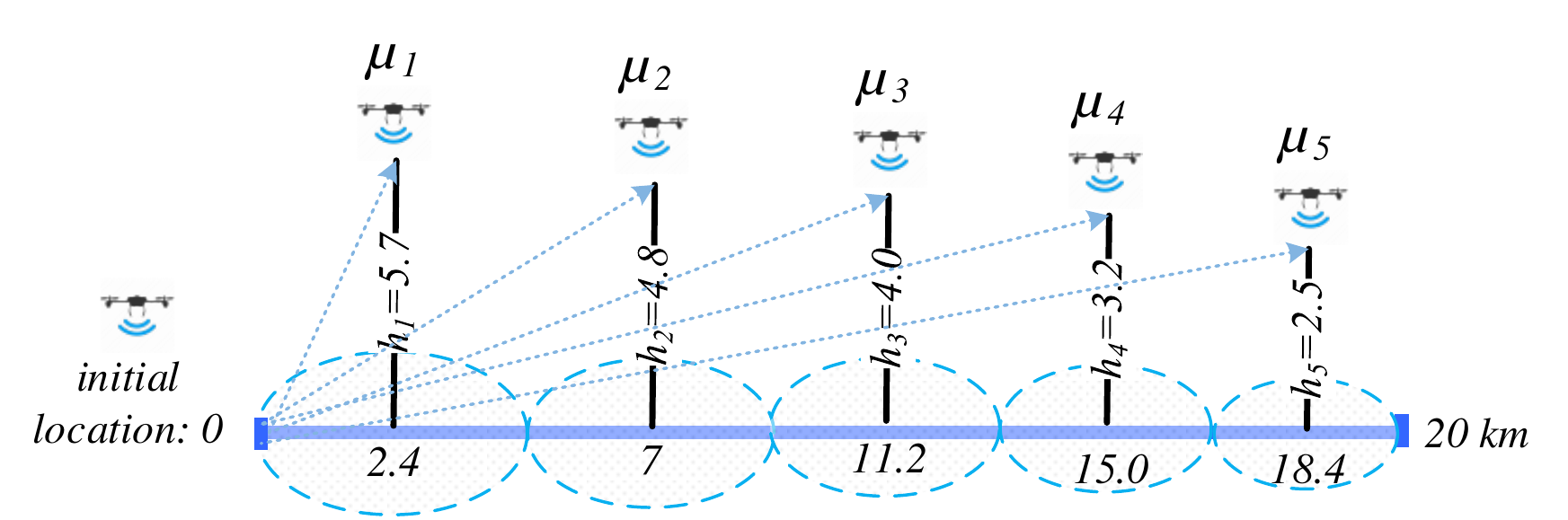}}
\hspace{-0.00in}
\subfigure[Same initial energy storage and with the consideration of NFZ.]{
\label{fig:subfig:ins2}
\includegraphics[width=0.7\textwidth]{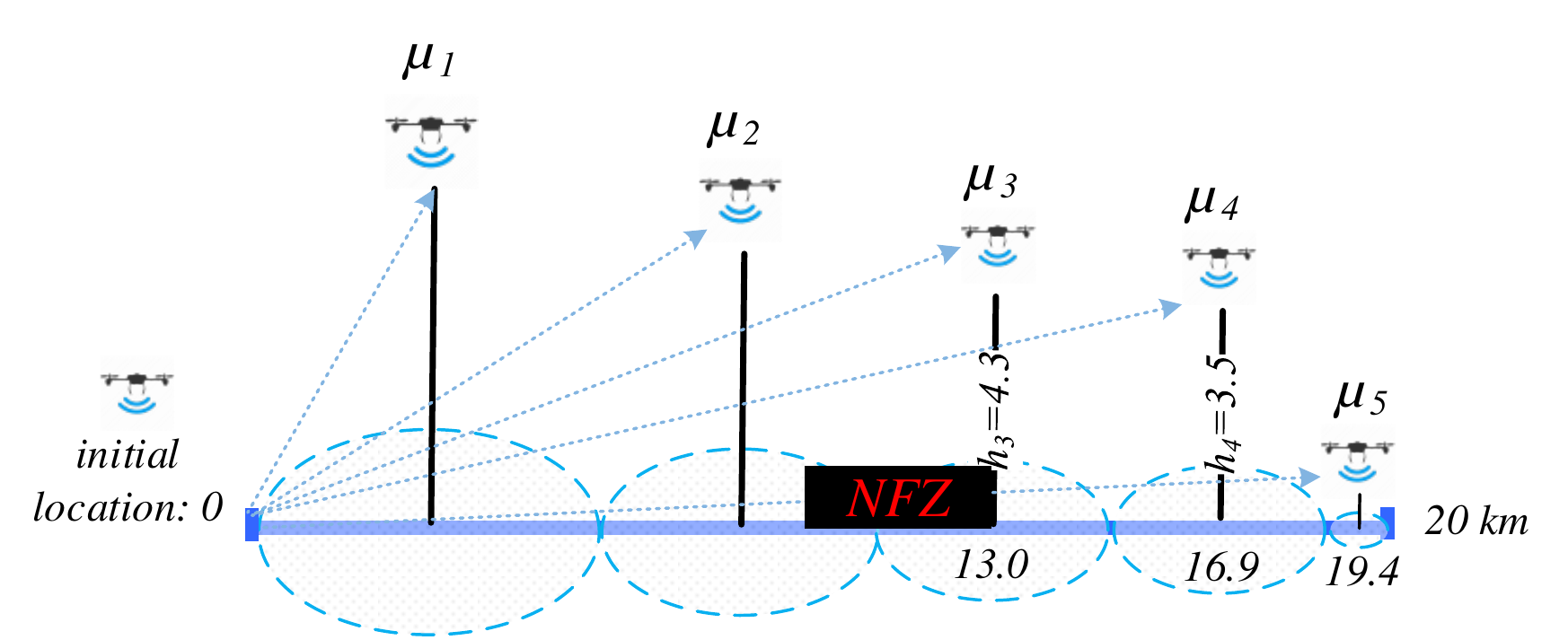}}
\hspace{-0.00in}
\subfigure[Increase initial energy storage $B_3 = 0.9~kWh$ for UAV $\mu_3$ and with the consideration of NFZ.]{
\label{fig:subfig:ins3}
\includegraphics[width=0.7\textwidth]{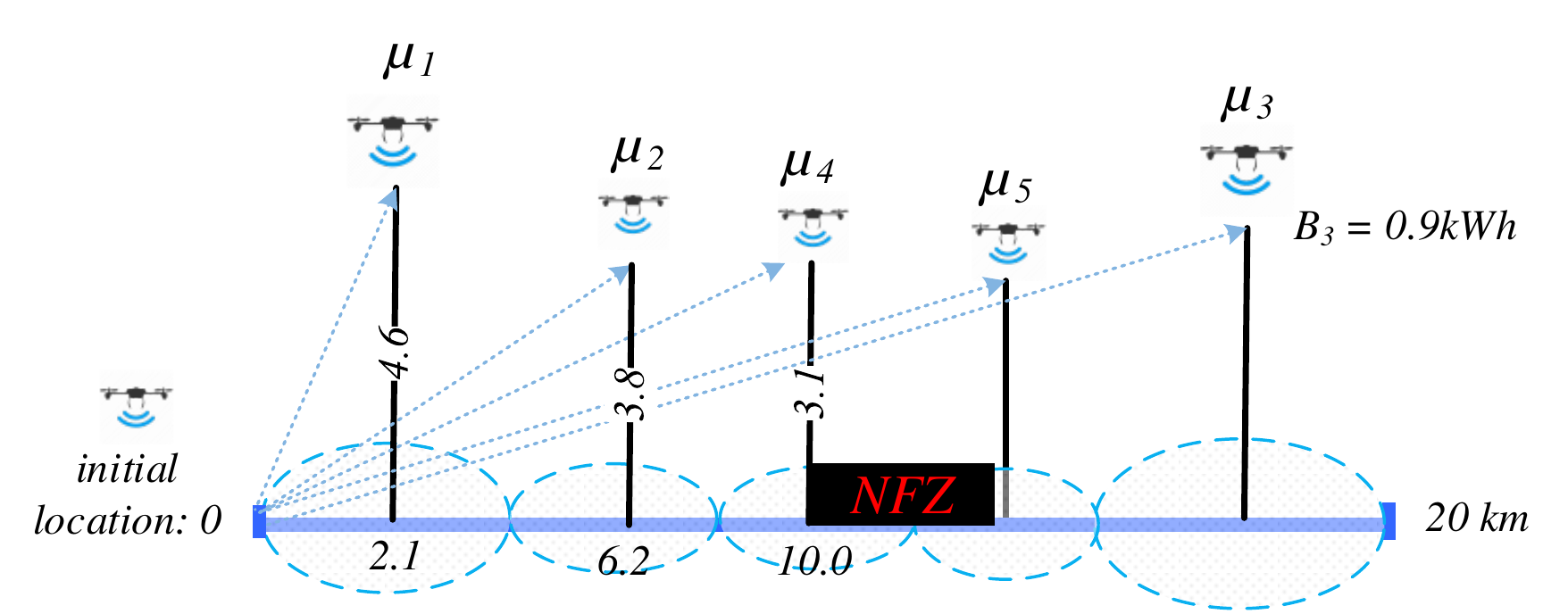}}
\caption{The optimal deployment solutions by dispatching 5 UAVs from the same initial location.} \label{fig:sins}
\end{figure}

\section{Simulations} \label{sec:exp}
In this section, we present extensive simulations to evaluate the proposed algorithms in different scenarios. Regarding the setting of UAVs' technical specifications, we approximate the wireless coverage radius of UAV $\mu_i$ hovering at altitude $h_i$ as a concavely increasing function $r(h_i) = \alpha h_i^{\beta}$ with $\alpha =1, \beta = 0.5$ till the turning point $h^*= 2km$ according to the measurement experiment done in \cite{al2014optimal}. We set $w = 0.2$, the initial energy storage $B = 0.78~kWh$, and the energy consumption per flying distance $c = 21.6~Wh/km$, as recommended by \cite{figliozzi2017drones}. Moreover, we set the length of target interval as $L = 20$ km, and the length of the NFZ as $3~km$ inside.


\begin{figure}[t]
    \centering
        \includegraphics[width=0.6\textwidth]{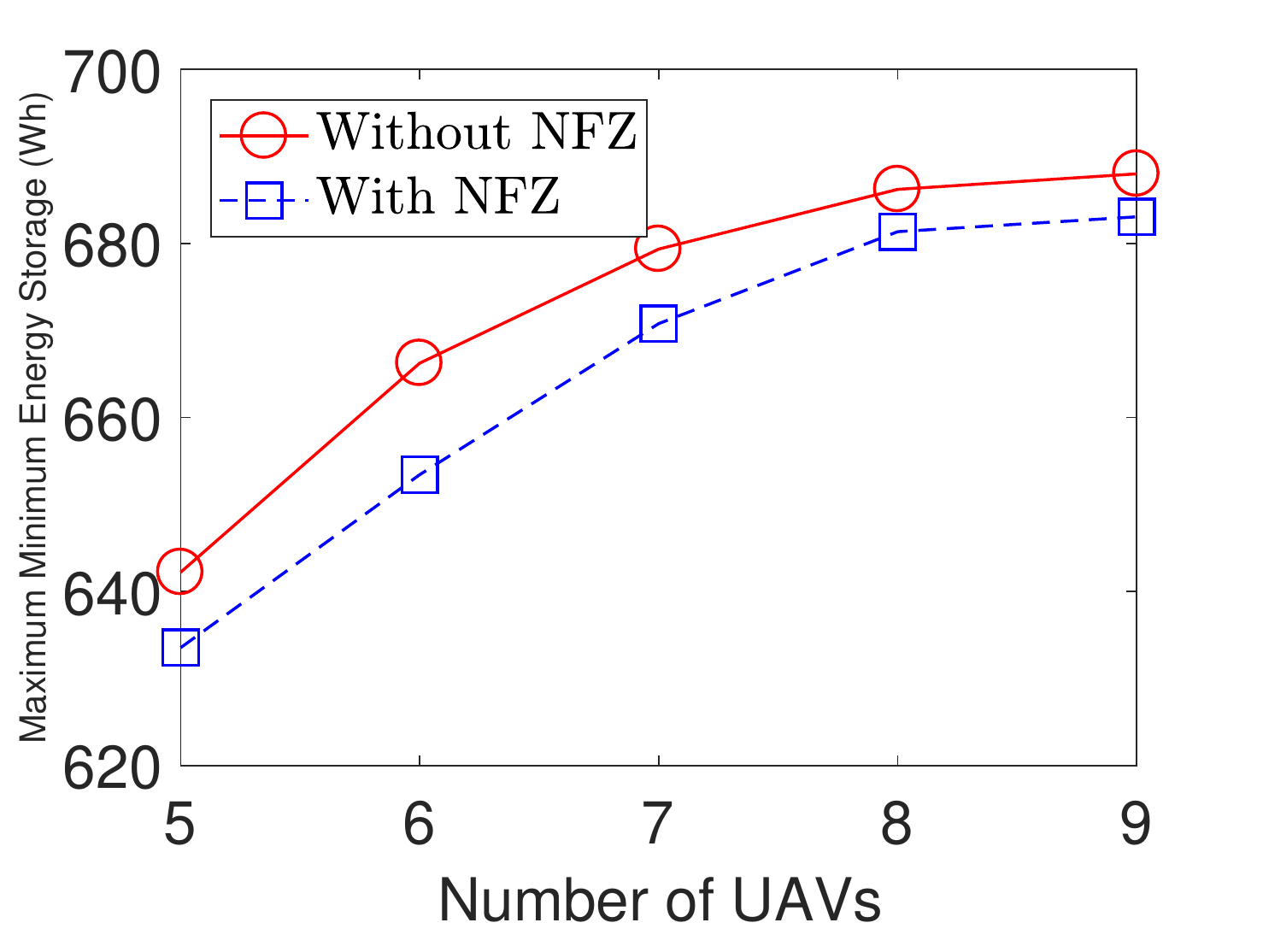}
    \caption{The maximum minimum leftover energy storage with and without NFZ consideration versus the number of UAVs.}
    \label{fig:sameo}
\end{figure}

\subsection{Deploying UAVs from the same station}\label{sec:expsame}

We first present the simulation results when dispatching the UAVs from the same location, as studied in Section~\ref{sec:same}. Figure~\ref{fig:sins} shows 5 UAVs' final service altitudes and ground destinations by solving (\ref{eq:sum}) and (\ref{eq:eq}). Figure~\ref{fig:subfig:ins1} shows the optimal solution without NFZ consideration, and it tells us that a UAV deployed further away to the right-hand-side of the ground should be placed to a lower altitude for keeping multi-UAVs' energy consumptions or travel distances the same. Figure~\ref{fig:subfig:ins2} shows the case when we add an NFZ $[\delta^l, \delta^r]$ with edges $\delta^l=10 km, \delta^r=13 km$. We can see that UAV $\mu_3$ has to re-locate to the right NFZ edge $\delta^r$. This pushes UAVs $\mu_4$ and $\mu_5$ further away to the right corner of $[0, L]$. Further, after increasing UAV $\mu_3$'s initial energy storage from $0.78$ to $0.9$ kWh in Figure~\ref{fig:subfig:ins3}, UAV $\mu_3$ with the largest energy storage flies to the right-most ground location and UAV $\mu_4$ re-locates to $\delta^l$ finally. This is consistent with Proposition~\ref{border}.

Next, we show the minimum leftover energy storage or UAV the swarm's residual lifetime (i.e., objective of problem (\ref{general})) as a function of the number of UAVs. Figure~\ref{fig:sameo} shows the minimum leftover energy storage increases as more UAVs are given, as each UAV only needs to cover a smaller range and this improves the bottleneck UAV performance. It also shows more energy will be consumed when we take into account the NFZ constraints.

\begin{figure}[t]
    \centering
    \includegraphics[width=0.7\textwidth]{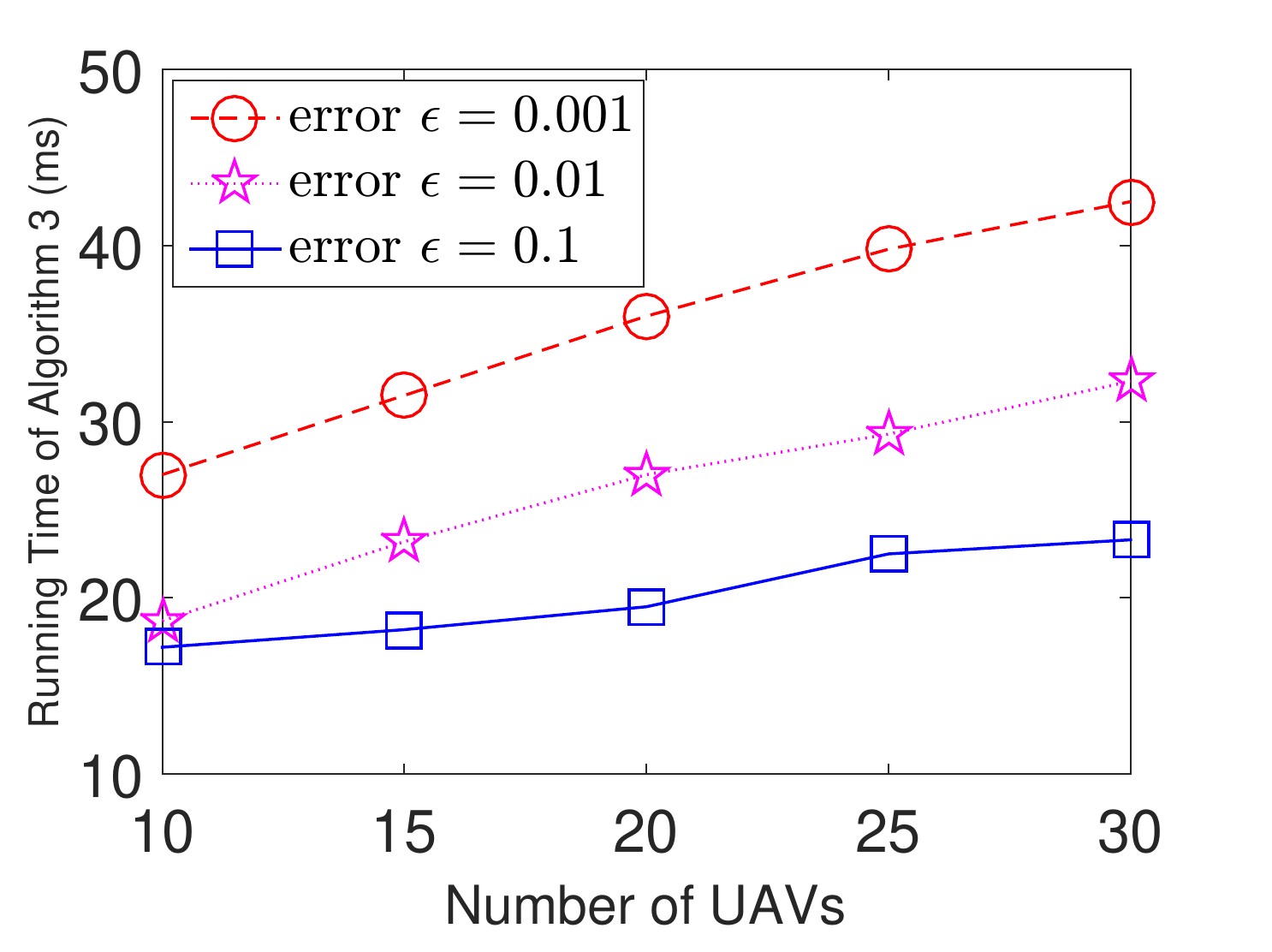}
    \caption{The running times (in milliseconds) of Algorithm~\ref{alg:fptas} under different values of $\epsilon$ and UAV number.}
\label{fig:time}
\end{figure}

\subsection{Deploying UAVs from different initial locations}\label{sec:expdif}
In Figure~\ref{fig:time}, we show the running time or computational complexity of our Algorithm~\ref{alg:fptas} in Section~\ref{sec:different} under different values of relative error $\epsilon = 0.1$, $0.01$ and $0.001$, respectively. All the UAVs have the same initial energy storage here. It can be observed that the smaller value of $\epsilon$ is, the longer running time is required. In addition, as the number of the UAVs increases, the running time increases yet is actually less than the theoretical bound $O(n \log \frac{1}{\epsilon})$ proved in Theorem~\ref{thrm:FPTAS}, since the theoretical analysis there is based on the worst case analysis and here is based on average or empirical analysis.

\begin{figure}[t]
    \centering
    \includegraphics[width=0.6\textwidth]{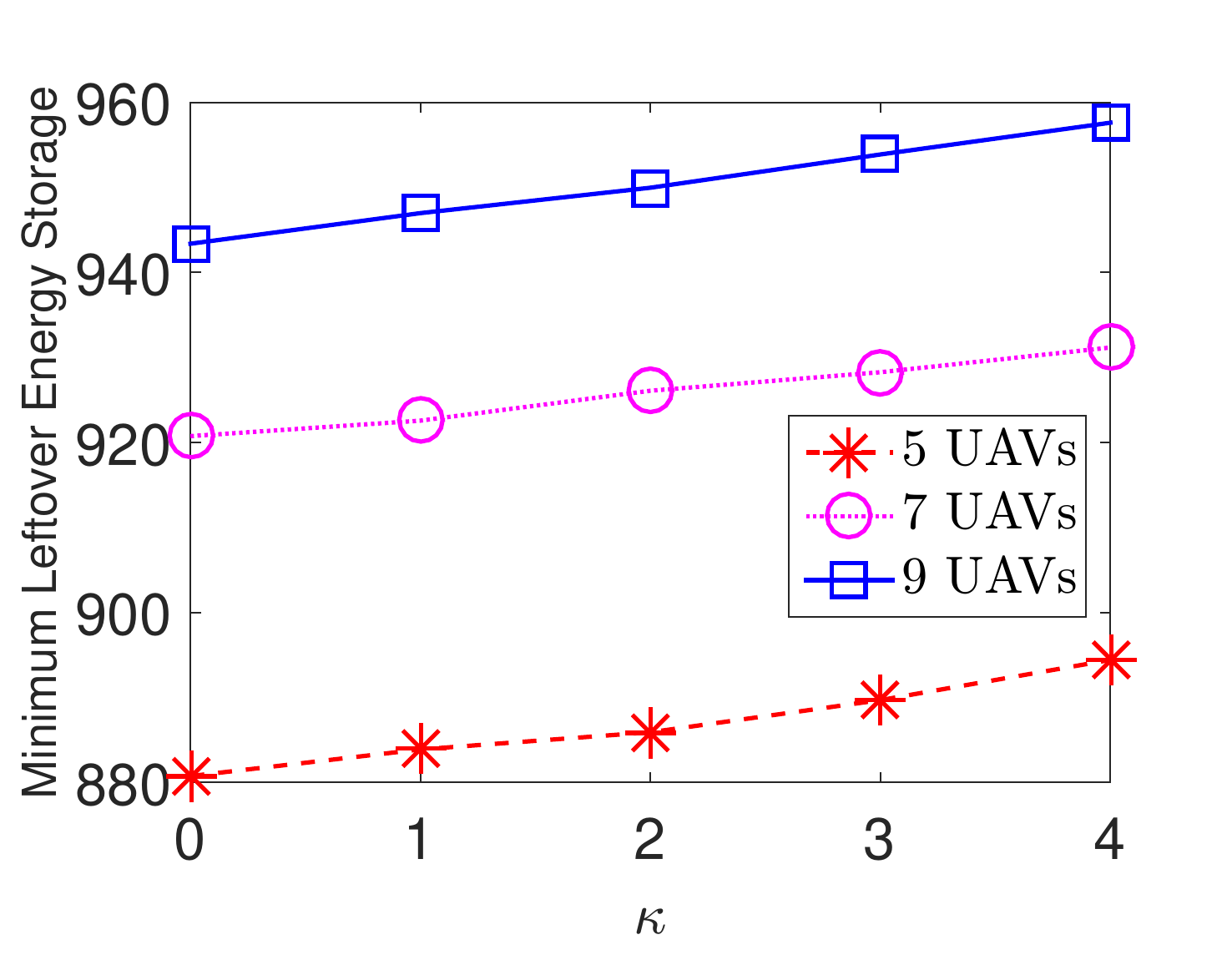}
    \caption{UAV swarm lifetime returned by Algorithm~\ref{alg:Heuristic} with different values of $\kappa$ for the UAVs with different initial energy storages.}
\label{fig:heu}
\end{figure}

Now we are ready to evaluate Algorithm~\ref{alg:Heuristic} in Section~\ref{sec:diffb} for deploying the UAVs with different initial locations as well as different energy storages. Recall that the initial order preserving property no longer holds, and we need to update the order changing degree $\kappa$ to balance the efficiency and complexity. In Figure~\ref{fig:heu}, problem (\ref{general}) is solved by Algorithm~\ref{alg:Heuristic} with $\kappa = 0$ (preserving all the UAVs' initial location order as in Algorithm~\ref{alg:fptas}), and with $\kappa = 1, 2, 3, 4$ (with increasingly more freedom to rotate any two neighboring UAVs's final ground destinations), respectively. The computational complexity is positively related to $\kappa$. As $\kappa$ increases, we gradually relax the location order preserving constraint, and Figure~\ref{fig:heu} shows the minimum leftover storages obtained by Algorithm~\ref{alg:Heuristic} increases.

\subsection{Deploying UAVs in the 3D space}\label{sec:two-sides}
In Figure~\ref{fig:twosides}, the sustainable UAV deployment problem in 3D is solved by Algorithm~\ref{alg:FPTAS2}, assuming all UAVs are dispatching from the two UAV stations on the left and right-hand corners of the target interval. As the UAV set is $\textbf{U} = U_l \cup U_r = \{\mu_1, \mu_2, \ldots, \mu_n\}$, we set $n = 10$ and discuss the effect of the difference between $|U_l|$ and $|U_r|$ on the UAV swarm's lifetime performance. By setting a larger value of $|U_l|$, we have more UAVs in the left station rather than the right station. We can see from Figure~\ref{fig:twosides}, when the division of UAVs is symmetric between the two UAV stations, our algorithm can achieve the longest UAV swarm lifetime, since the UAVs on the left-hand-side (or right-hand-side) station do not need to fly further to cover the right (left) part of the target interval $[0, L]$.

\begin{figure}[tbh]
    \centering
    \includegraphics[width=0.7\textwidth]{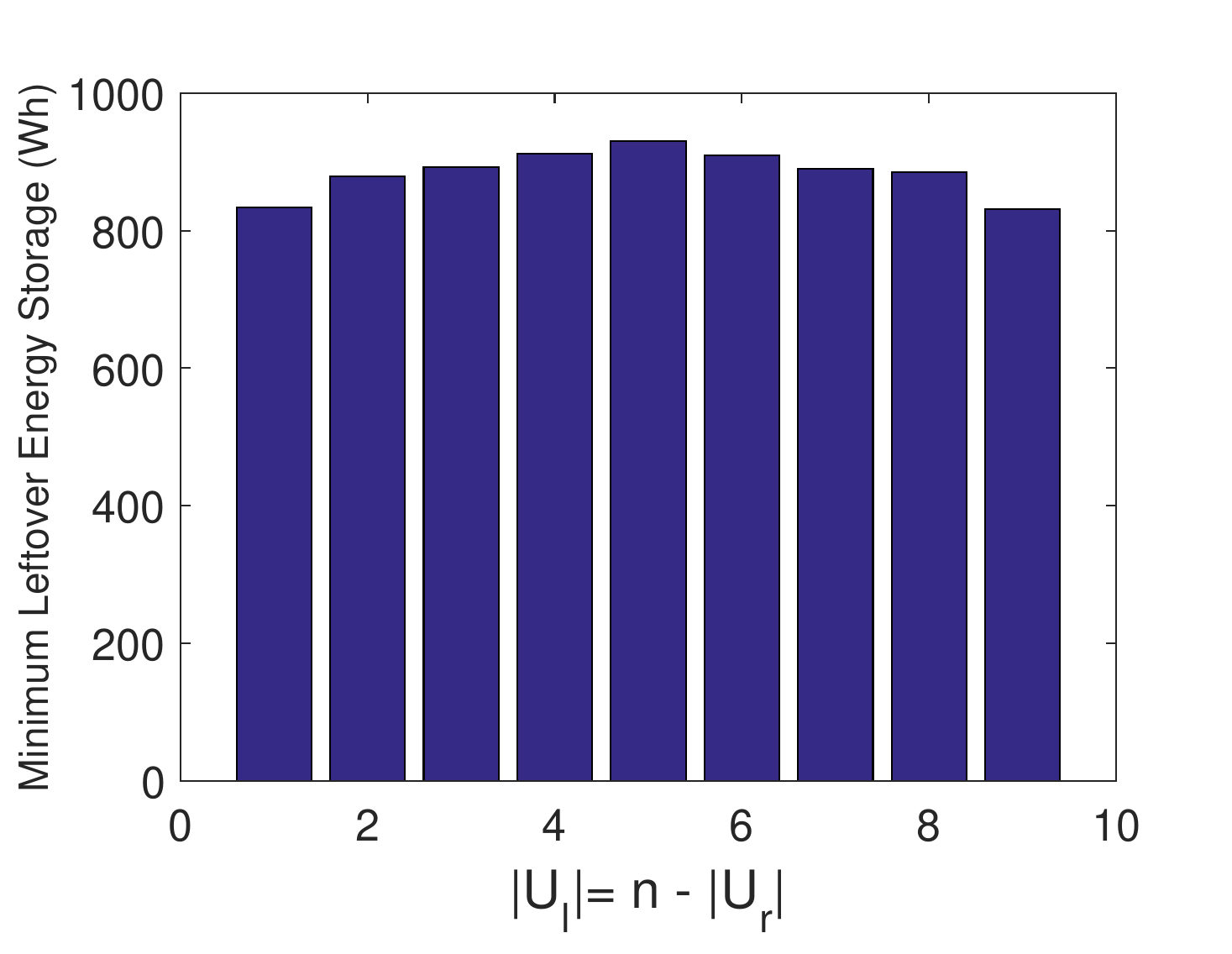}
    \caption{Output of Algorithm~\ref{alg:FPTAS2} with different UAV distributions between the two UAV stations.}
\label{fig:twosides}
\end{figure}

\section{Conclusion} \label{sec:dis}

The energy-saving UAV deployment problem for keeping wireless coverage is of great practical importance. We study the energy-saving UAV deployment problem to prolong the UAV swarm's residual lifetime. When all UAVs are deployed from a common UAV station, we propose an optimal deployment algorithm, by jointly optimizing UAVs' flying distances on the ground and final service altitudes in the sky. We show that a UAV with larger initial energy storage should be deployed further away from the UAV station for balancing multi-UAVs' energy consumption in the flight. Due to NFZs consideration, the problem becomes more difficult and the whole UAV swarm consumes more energy. We solve it optimally in $O(n \log n)$ time. Moreover, when $n$ UAVs are dispatched from different initial locations, we first prove that any two UAVs will not fly across each other in the flight as long as they have the same initial energy storage, and then design an approximation algorithm to arbitrarily approach the optimum. Further, we consider that UAVs may have different initial energy storages under the constraint of NFZs, and we prove this problem is NP-hard. Despite of this, we successfully propose a heuristic algorithm to solve it by balancing the efficiency and computation complexity well. Finally, we extend the FPTAS to a 3D scenario and validate theoretical results by extensive simulations. Still, there are some open issues in future work, i.e., the interference among UAVs' wireless services and/or the algorithms' extensions to cover higher dimensional target, e.g 2D area or 3D object.

\clearpage
\appendices

\section{Proof of Proposition~\ref{seamless}} \label{ap:att}

\begin{figure}[t]
    \centering
        \includegraphics[width=0.7\textwidth]{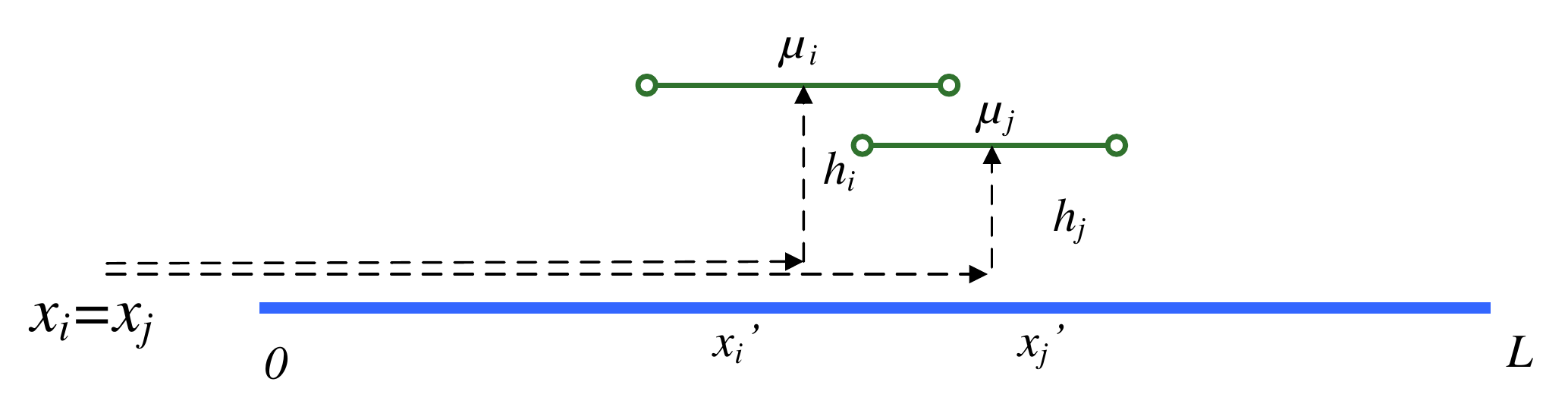}
    \caption{Two neighboring UAVs' coverages overlap.}
    \label{fig:overlap}
\end{figure}

\begin{proof}
As shown in Section \ref{sec:sys}, for energy saving, at the optimum the altitude $h_i$ of each UAV $\mu_i$ after deployment to be not higher than $h^*$. Thus, we only consider the concavely increasing part of the function $r(h_i)$ here. First, we show that any two neighboring UAVs will not overlap in the optimal solution by contradiction. Suppose in the optimal solution, the coverages of two UAVs $\mu_i$ and $\mu_j$ overlap as shown in Figure \ref{fig:overlap}. In this case, we can simply lower the service altitude of $\mu_i$ and move left relatively to remove the overlap, while decrease its energy consumption for moving. We can see that we can obtain better solution and the leftover energy storages of both UAVs do not increase by removing the overlap. This is a contradiction. Next, suppose two UAVs $\mu_i$ and $\mu_{j}$ seamlessly cover a subinterval of $L$ in the optimal solution, but they have different leftover energy (i.e., $\hat{B}_i = B_i-c \cdot d_i$, $\hat{B}_j = B_{j} - c \cdot d_{j}$). If $\hat{B}_i < \hat{B}_j$ and $\hat{B}_i$ is the bottleneck, we can lower the service altitude of $\mu_i$ and move left relatively to decrease its energy consumption while increase the service altitude of $\mu_{j}$ for keeping exactly the same total coverage. In this way, the bottleneck $\mu_i$ has more leftover energy and the objective of problem (\ref{general}) is further improved. Otherwise, $\hat{B}_i > \hat{B}_j$ and $\hat{B}_j$ is the bottleneck, we can lower the service altitude of $\mu_j$ and move right relatively to decrease its energy consumption. In the meantime, we increase the service altitude of $\mu_i$ and move right to cover more for keeping exactly the same total coverage. In the end, the bottleneck $\hat{B}_j$ is increased and our proof is completed.
\end{proof}

\section{Proof of Theorem~\ref{thry:np}} \label{ap:npc}
\begin{proof}

Given an instance of the sustainable UAV deployment problem and $\hat{B}$, we show that it is NP-hard to determine whether the UAVs' leftover energy storage is at least $\hat{B}$ after deployment. Here, we choose concave function $r(h_i) = \alpha h_i^{\beta}$ till turning point $h^*$ and set $\alpha = 0.5$ and $\beta = 1$, and the problem is NP-hard as long as we show this special problem is NP-hard.

We reduce the {\em Partition problem}, which is a well-known NP-hard problem \cite{garey2002computers}, to our sustainable UAV deployment problem. The Partition problem is defined as follows: Given a sequence of positive integers $1 \leq a_{1} \leq a_{2} \leq, \ldots, \leq a_{m}$, we want to determine whether there exists a set of indices $\Gamma \subseteq \{1, 2, \ldots, m\}$ such that $\sum_{i \in \Gamma}{a_{i}} = \frac{1}{2} \sum_{i=1}^{m}{a_{i}}$.

Given a Partition instance, we construct the following sustainable UAV deployment problem. There is a UAV for each input number: $x_i = \frac{L}{2}$, and $B_1 = \rho + 3r(h_1) + \hat{B}$, $B_n = \frac{L}{2} + r(h_n) + \hat{B}$, and $B_{1} \leq B_{2} \leq, \ldots, \leq B_{n}$. We add one UAV $\mu_{m+1}$ with $B_{m+1} = 2\rho + \hat{B}$. $r(h_1) > \rho$.

We show that if $a_{1}, \ldots, a_{m} \in $ Partition, then there exists a solution with $\hat{B}$. $\sum_{i \in \Gamma}{a_{i}} = \frac{1}{2} \sum_{i=1}^{m}{a_{i}} = \frac{L}{2} - \rho$. We set $r_i = \frac{1}{2} a_i$. We deploy the UAVs with initial energy storage $B_i$ for $i \in \Gamma$ to cover $[0, \frac{L}{2} - \rho]$.

For UAV $\mu_i$ ($i \in \Gamma$), assume $|\Gamma| = T$, we make the initial energy storage $B_i$ satisfy the following requirements:

\begin{align}
&\frac{L}{2} + r(h_1) = B_1 - \hat{B} \notag \\
&\ldots \notag \\
&\frac{L}{2} - 2(r(h_1) + \ldots + r(h_{T-1}) + r(h_T) = B_T - \hat{B} \notag \\
&r(h_1) + \dots + r(h_{T-1}) + r(h_T) = \frac{L}{2} - \rho \notag
\end{align}
Similar process can be conducted on UAVs $\mu_i$ ($i \notin  \Gamma$) to cover $[\rho, \frac{L}{2}]$. We can compute $\sum B_i = \frac{1}{2} \sum_{i=1}^{m}{a_{i}}$ to just fully cover $[0, \frac{L}{2} - \rho]$ or $[\frac{L}{2} + \rho, \frac{L}{2}]$.


In our sustainable UAV deployment problem, we have $n=m+1$ UAVs which are initially located at the position $\frac{L}{2}$. Moreover, $\mu_i$ is associated with initial energy storage $B_i = a_i$ for $1 \leq i \leq m$ and $\mu_{m+1}$ is associated with $2 \rho + \hat{B}$. There are two NFZs with length $2 \rho$ located at both sides of point $\frac{L}{2}$. This transformation can clearly be performed in polynomial time. See Figure~\ref{fig:np} for an example.

\begin{figure}[tbh]
    \centering
        \includegraphics[width=0.7\textwidth]{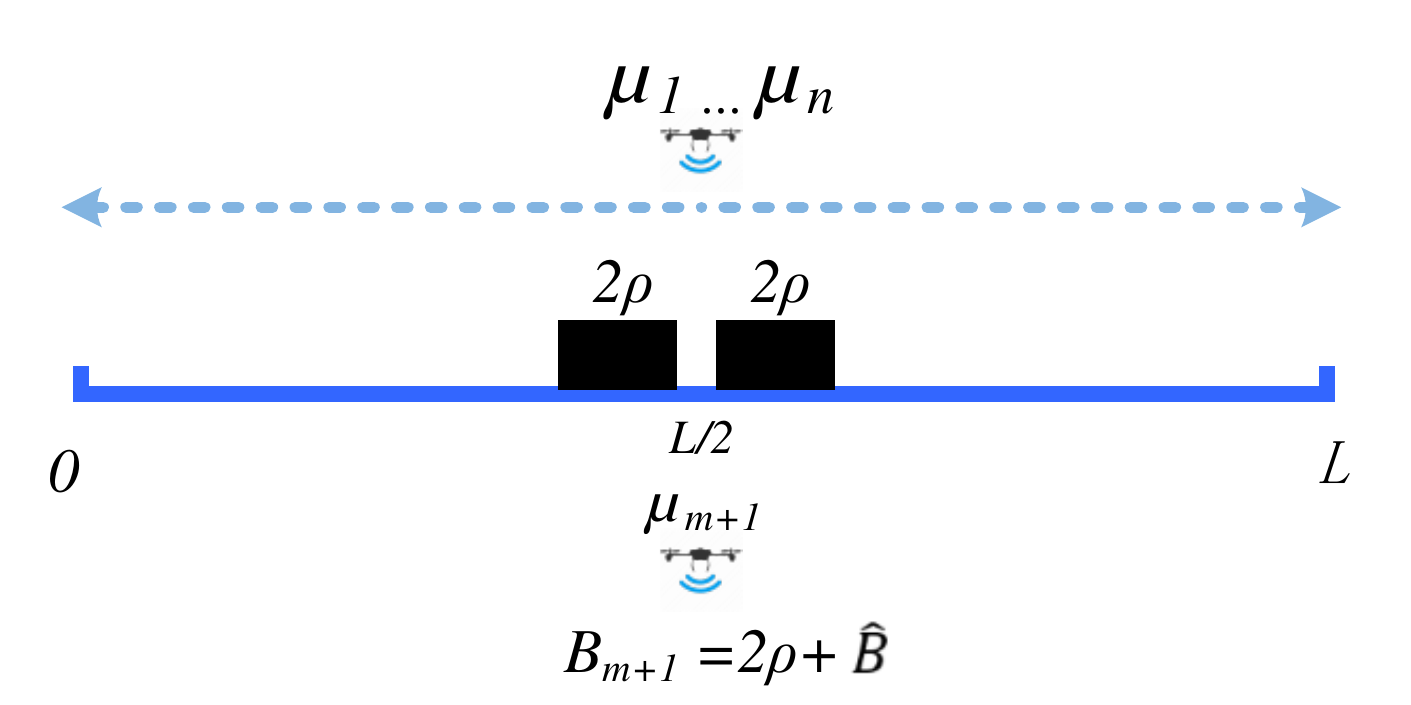}
    \caption{NP-hard proof of sustainable UAV deployment problem.}
    \label{fig:np}
\end{figure}

We now prove that there exists a solution to the instance $P$ of the Partition problem if and only if the constructed instance $C$ of the sustainable UAV deployment problem has a solution of at least $\hat{B}$.\\

\noindent $(\boldsymbol \Longrightarrow)$ Given a solution $\Gamma$ to the instance $P$ of the Partition problem, we can move UAVs $\mu_i$ for $i \in \Gamma$ to the left such that every point of the interval $[0,\frac{L}{2} - \rho]$ is covered and move $\mu_j$ for $j \not \in \Gamma$ to the right such that every point of the interval $[\frac{L}{2} + \rho, \frac{L}{2}]$ is covered. Moreover, we move $\mu_{m+1}$'s vertically to $\rho$ such that every point of the interval $[\frac{L}{2} - \rho, \frac{L}{2} + \rho]$ is covered. Since we have $\sum_{i \in \Gamma}{2 \cdot r(h_i)} =\sum_{j \in \{1, 2, \ldots, m\} \setminus \Gamma }{2 \cdot r(h_j)} = \frac{1}{2} \sum_{i=1}^{m}{a_{i}} = \frac{L}{2} - \rho$, it is obvious to see that this is a feasible solution and the detailed movements can be implemented in the straightforward way. Therefore, we have a solution of at least $\hat{B}$ to the instance $C$ of the sustainable UAV deployment problem.\\

\noindent $(\boldsymbol \Longleftarrow)$ Now we have a feasible solution of at least $\hat{B}$ to the instance $C$ of the sustainable UAV deployment problem. We first observe that the UAV $\mu_{m+1}$ has to fly vertically to $2\rho$. Since every point of the interval $[0, L]$ is required to be covered, there does not exist an ``overlapped interval'' between any two UAVs' ranges in a feasible solution. For the UAVs $\mu_1, \ldots, \mu_n$, we need to find a subset of UAVs with the summation of battery to fully cover $[0,\frac{L}{2} - \rho]$ or $[\frac{L}{2} + \rho, \frac{L}{2}]$, which is $\frac{1}{2} \sum_{i=1}^{m}{a_{i}}$. The above analysis shows that $\sum_{i \in \Gamma} r(h_i)= \sum_{i \in \Gamma}{a_{i}}=\frac{1}{2} \sum_{i=1}^{m}{a_{i}}$, which implies that we have a solution to the instance $P$ of the Partition problem. The proof is thus complete.
\end{proof}

\section{Proof of Lemma~\ref{biorder}} \label{ap:biorder}
\begin{proof}
First of all, for UAV subset $U_l$ or $U_r$ in each of the left and right-hand side UAV stations, Proposition~\ref{border} still holds. That is, we should deploy all the UAVs in $U_l$ from left to right hand side of the target interval according to increasing order of by their initial energy storages. Similarly, we deploy the UAVs in $U_r$ from right to left hand side of the target interval according to increasing order of by their initial energy storages. At the optimum, all the UAVs' leftover energy is the same.

Now, we are ready to prove this lemma by contradiction. Consider that there are two UAVs $\mu_i$ and $\mu_j$ that will cross each other with $x_i^{\prime} \geq x_j^{\prime}$ and become neighbors after deployment. Here, $\mu_i$ is originally located in $U_l$ and $\mu_j$ in $U_r$ (i.e., $x_i = x_l$, $x_j = x_r$). They cover a continuous line interval $[x_j^{\prime}-\sqrt{r(h_j)^2- {y_j^{\prime}}^2}, x_i^{\prime}+\sqrt{r(h_i)^2-{y_i^{\prime}}^2}]$ along target interval $[0, L]$, and they have the same leftover energy ($B_i - d_i = B_j - d_j = \hat{B}^*$) in the optimal solution. Next we swap $\mu_i$ and $\mu_j$ to show a better solution is actually achieved. If we swap the ground location order of $\mu_i$ and $\mu_j$ without changing any UAV's coverage radius or altitude, we can see that $d_i^{\prime} = w \cdot \sqrt{(x_l-x_i^{\prime})^2 + (y_l-y_i^{\prime})^2}+h_i$ and $d_j^{\prime} =  w \cdot \sqrt{(x_r-x_j^{\prime})^2+ (y_r-y_j^{\prime})^2}+h_i$ will decrease, i.e., $d_i^{\prime} < d_i$ and $d_j^{\prime}< d_j$. Specifically, we move $\mu_i$ to cover from $x_j^{\prime} - \sqrt{r(h_j)^2-{y_j^\prime}^2}$ and move $\mu_j$ to cover until $x_i^{\prime} + \sqrt{r(h_i)^2-{y_i^\prime}^2}$. We can see that the original covered range $[x_j^{\prime} -\sqrt{r(h_j)^2- {y_j^{\prime}}^2}, x_i^{\prime} + \sqrt{r(h_i)^2- {y_i^{\prime}}^2}]$ is still covered by $\mu_i$ and $\mu_j$, and $d_i^{\prime} < d_i$ and $d_j^{\prime}< d_j$. Given $B_i - c d_i^{\prime}  > \hat{B}^*$ and $B_j - c d_j^{\prime} > \hat{B}^*$, a better solution is achieved to prolong the UAV network lifetime. The proof is completed.
\end{proof}


\bibliographystyle{plainnat}

\end{document}